\documentclass{article}
\usepackage[journal=JST]{ems-journal}
\usepackage{svg}
\usepackage{caption}
\usepackage{subcaption}

\usepackage{amsmath,amsthm,amscd,amsfonts}
\usepackage{xcolor}
\usepackage{eepic}
\usepackage{enumitem}

\usepackage[PostScript=dvips,balance,midshaft,nohug]{diagrams}
\newarrow {Line} {-}{-}{-}{-}{-}
\newarrow{Into}{C}{-}{-}{-}{>}

\swapnumbers  

\newtheorem{thm}[subsection]{Theorem}
\newtheorem{cor}[subsection]{Corollary}
\newtheorem{lem}[subsection]{Lemma}
\newtheorem{translem}[subsection]{Transversality Lemma}
\newtheorem{prop}[subsection]{Proposition}
\theoremstyle{definition}

\newtheorem{rems}[subsection]{Remarks}

\numberwithin{equation}{section}

\renewcommand{\H}{\mathcal H}
\newcommand{\RR}{\mathbb R}
\newcommand{\CC}{\mathbb C}
\renewcommand\mod[1]{\ (\mathop{\mathrm mod}#1)}

\newcommand{\ZZ}{\mathbb Z}
\newcommand{\TT}{\mathbb T}

\newcommand{\diag}{\mathop {\mathrm{diag}}}
\newcommand{\ind}{\mathrm{ind}}
\newcommand{\Hess}{\mathrm{Hess}}

\newcommand{\spt}{\mathrm{spt}}
\newcommand{\sslash}{\mathbin{/\mkern-6mu/}}

\usepackage[utf8]{inputenc}

\begin{document}
	
	\title{Morse theory for discrete magnetic operators and nodal count distribution for graphs}
 \titlemark{Morse theory and nodal count}
\emsauthor{1}{Lior Alon}{L. Alon}
\emsaffil{1}{Department of Mathematics, Massachusetts Institute of Technology, Cambridge MA 02139 \email{lioralon@mit.edu}}
\emsauthor{2}{Mark Goresky}{M. Goresky}
\emsaffil{2}{School of Mathematics, Institute for Advanced Study, Princeton NJ 08540 
\email{goresky@ias.edu}}


\begin{abstract}
Given a discrete Schrödinger operator $h$ on a finite connected graph $G$ of $n$ vertices,
the nodal count $\phi(h,k)$ denotes the number of edges on which the $k$-th eigenvector changes sign. A {\em signing} $h'$ of $h$ is any real symmetric matrix constructed by changing the sign of some off-diagonal entries of $h$, and its nodal count is defined according to the signing. The set of signings of $h$ lie in a naturally defined torus $\TT_h$ of ``magnetic perturbations" of $h$. G. Berkolaiko \cite{ Berkolaiko} discovered that every signing 
$h'$ of $h$ is a critical point  of every eigenvalue $\lambda_k:\TT_h \to \RR$, with Morse index 
equal to the nodal surplus. We add further Morse theoretic information to this result. 
We show if $h_{\alpha} \in \TT_h$ is a critical point of $\lambda_k$ and the eigenvector vanishes at a 
single vertex $v$ of degree $d$, then the critical point 
lies in a nondegenerate critical submanifold of dimension $d+n-4$, closely related to the configuration
space of a planar linkage.  We  compute its Morse index in terms of spectral data. 


 The {\em average nodal surplus distribution} 
is the distribution of values of $\phi(h',k)-(k-1)$, averaged over all signings $h'$ of $h$. If all critical points correspond to simple eigenvalues with nowhere-vanishing eigenvectors, then the average nodal surplus distribution is binomial. In general, we conjecture that the nodal surplus distribution converges to a 
Gaussian in a CLT fashion as the
first Betti number of $G$ goes to infinity. 
\end{abstract}	

\maketitle
	
	\section{Introduction}
 
	In some ways this paper is both an analog of \cite{ABBuniversality} for discrete graphs and a continuation and expansion of the papers \cite{Berkolaiko,CdeV2}, although it is completely self-contained.

	\subsection{The Setting}\label{subsec-intro1}
	Let $G$ be  a simple graph on $n$ ordered vertices labeled $1, 2, \cdots, n$.  Write $r \sim s$ if $r\ne s$ are vertices connected by an edge.
	A (real or complex) {\em function} on $G$ is a function on the vertices of $G$, that is, a vector
	in $\RR^n$ or $\CC^n$ and we denote the value of such a function $v=(v_1,v_2,\cdots,v_n)$ by $v(r)$ or $v_r$.  
	An $n \times n$ matrix $h$ is {\em supported on $G$} if $h_{rs} \ne 0 \implies r\sim s \text{ or } r = s$.
	Let $\mathcal S(G)$ and $\mathcal H(G)$ denote the vector spaces of real symmetric matrices and complex Hermitian matrices supported on $G$. 
	A  {\em discrete Schr\"odinger operator} is a real symmetric matrix $h\in \mathcal S (G)$ with $h_{rs} < 0$ for $r \sim s$. 
	The quadratic form associated with $h\in \mathcal S(G)$ may be expressed as the quadratic form of $\Delta +V$, that is 
	\begin{equation}\label{quadratic form}
		\langle f, hf \rangle = -\sum_{r \sim s} h_{rs}\left(f(r) - f(s)\right)^2 + \sum_{r=1}^n V(r) f(r)^2    
	\end{equation}
	where the ``potential"  is $V(r) = h_{rr} + \sum_{r\sim s}h_{rs}$ and $\Delta$ is a {\em weighted Laplace operator} on $G$.

	A discrete Schr\"odinger operator $h$ has real eigenvalues $\lambda_1 \le \lambda_2 \le \cdots \le \lambda_n$.
	Suppose  $\lambda_k$ is a simple (multiplicity one) eigenvalue of $h$ with a nowhere-vanishing eigenvector $v$ (meaning that $v_r \ne 0$ for all $r$). 
	A basic problem in graph theory is to understand the behavior of the {\em nodal count}  $\phi(h,k)$, that is,
	the number of edges $r \sim s$ for which $v$ changes sign:  $v(r)v(s)< 0$.   It is known that
	\begin{equation}\label{nodal bound}
		k-1\le \phi(h,k)\le k-1+\beta,
	\end{equation}
	where $\beta$ is the first Betti number of $G$.  (See \cite{Davis} for a review of the many works leading to the upper bound, an analogue of
	Courant's theorem\footnote{The Courant theorem  states, for a domain
		$\Omega$ in Euclidean space with homogeneous boundary conditions, that the
		nodal set of the $k$-th eigenfunction of the Laplacian divides $\Omega$ into no more than $k$ subdomains, see \cite{Courant} Chapt. 6 \S 6.}, 
	and \cite{Berkolaiko3} for the lower bound.)
	This motivates the definition of the {\em nodal surplus} 
	\[\phi(h,k)-(k-1)\in \{0,1,\cdots,\beta\}\] and its probability distribution $P(h)=(P(h)_{0},\ldots,P(h)_{\beta})$ over the $n$ possible eigenvalues:
	\[P(h)_s = \frac{1}{n} \# \big\{1 \le k \le n\big|\ \phi(h,k)-(k-1) = s\big. \}.
	\]
	In numerical simulations for large graphs, this distribution seems to concentrate around $\frac{\beta}{2}$
	with variance of the order of $\beta$, similar to the observations for metric graphs in \cite{ABBuniversality}. 
	
	\subsection{Nodal count for signed graphs}
	If $h \in \mathcal S(G)$ is a discrete Schr\"odinger operator we may consider other {\em signings} $h' \in 
\mathcal S(G)$ obtained from $h$ by changing the sign of some collection of off-diagonal entries. Every
symmetric matrix $h' \in \mathcal S(G)$ is a signing of a uniquely determined Schr\"odinger operator $h$.
We may consider $h'$ to be an analog of the discrete Schr\"odinger operator on the corresponding {\em signed graph} $G'$ obtained from
$G$ by attaching signs to the edges, as originally introduced in \cite{Harary} and extensively 
studied, see \cite{bilu2006lifts, marcus2013interlacing, Zaslavsky}. In this case, taking the signing into account, the nodal count is
defined to be the number of edges $r \sim s$ such that $v(r)h'_{rs}v(s)>0$.

 Denote by $\mathcal S(h)$ the collection of all possible signings of $h$ (cf.~\S 2.6). The inequality (\ref{nodal bound}) continues to hold for any signing of $h$. 
The {\em average nodal surplus distribution} $P(\mathcal{S}(h))$ is the average of $P(h')$ over all signings $h'\in \mathcal{S}(h)$. In Theorem \ref{thm-index} we show that if the
diagonal entries of $h$ are all equal, then $P(\mathcal S(h))$ is symmetric around $\beta/2$.  Numerical experiments lead to the following
\paragraph{\textbf{Conjecture}}
		Given a simple connected graph $ G $ there is a generic set (open, dense and full measure) of $ h\in\mathcal{S}(G) $ for which the average nodal surplus distribution $P(\mathcal{S}(h))$ is symmetric around $ \beta/2 $ with variance $ \sigma_{h}^{2} $ of order $ \beta $. Moreover, the normalized distribution
		\[\rho_{G,h}:=\sum_{j=0}^{\beta}P(\mathcal{S}(h))_{j}\delta_{x_{j}}\quad\text{with}\quad x_{j}=\frac{j-\beta/2}{\sigma_{h}}, \]
		converges in the weak topology to the normal Gaussian distribution $ N(0,1) $ as $ \beta\to\infty $, uniformly over all simple connected $ G $ with first Betti number $ \beta $, and generic $ h\in\mathcal{S}(G) $.

\subsection{Gauge invariance}  The {\em gauge group} $\TT^{n}=(\RR/2\pi\ZZ)^n$ acts on the space $\mathcal H(G)$  where $(\theta_1,\theta_2,\allowbreak \cdots,\theta_n)$ acts by conjugation with $\diag(e^{i\theta_1}, e^{i\theta_2}, \cdots,
e^{i\theta_n})$.  This action preserves eigenvalues, nodal count, and most other graph properties that are studied in this paper.	Elements $h, h' \in \mathcal H(G)$ that differ by a gauge transformation are said to be {\em gauge equivalent}.  
If $h \in \mathcal S(G)$ is a discrete Schr\"odinger operator then the signings $h'$ of $h$ for which the corresponding signed 
graph $G'$ is {\em balanced} (see \cite{Harary}) are exactly those $h'$ 
that are gauge equivalent to $h$. 
	
\subsection{Magnetic operators and nodal count}
	In \cite{Berkolaiko, BerkolaikoWeyand} G. Berkolaiko suggested that one might better understand the nodal count by considering its variation under
	{\em magnetic perturbations of h}.  The discrete analog for the Schr\"odinger operator associated to
	a particle in a magnetic field appears in \cite{Harper1, Harper2}. See also \cite{Lieb}, \cite {CdeVMagnetic}, \cite[\S 2.1]{CdeV1} and \cite{CdeV2}.  It is quickly reviewed in Appendix \ref{sec-magnetic-Schrodinger}.
	
	Given a discrete Schr\"odinger operator $h\in \mathcal{S}(G)$, a magnetic potential $\alpha$ is a real anti-symmetric matrix supported on $G$ and the associated {\em magnetic Schr\"odinger operator} $h_{\alpha}\in\mathcal{H}(G)$ is the Hermitian matrix $(h_{\alpha})_{rs}=e^{i\alpha_{rs}}h_{rs}$.
	The manifold  ({\ref{eqn-TAh}}) of such magnetic perturbations, $\TT_h \subset \mathcal H_n$, is a 
	torus containing $h$, cf.~\S {\ref{subsec-star-action}} below.  Its quotient, see equation 
	(\ref{eqn-magnetic-manifold}), modulo gauge transformations, $\mathcal M_h$ is a torus 
	of dimension $\beta$.  In \cite{CdeV2} and \cite{Berkolaiko}, G. Berkolaiko and Y. Colin de Verdi\`ere discovered a remarkable fact:  for any real symmetric $h\in \mathcal{S}(G)$ with simple eigenvalue $\lambda_k$ and nowhere vanishing eigenvector, the nodal surplus $\phi(h,k) - (k-1)$ is
	equal to the Morse index of $\lambda_k$, interpreted as a Morse function on the manifold $\mathcal M_h$. 
	
\subsection{Morse theory for magnetic perturbations modulo gauge transformations}
	We wish to apply Morse theory to the function $\lambda_k:\mathcal H_n \to \RR$,
	restricted to the torus $\TT_h$ or its quotient $\mathcal M_h$.  In principle, Morse theory provides a
	prescription for building the homology of $\mathcal M_h$ from local data at the critical points of $\lambda_k$ together with some homological information as to how these local data fit together.  Since the homology of $\mathcal M_h$ is known, Morse theory should provide restrictions on the number and type of critical points of $\lambda_k$, and in turn, restrictions on the nodal surplus.

	There are several difficulties with this plan, the first being that $\lambda_k$ is continuous 
	but not smooth:  it is analytic on each stratum of a certain stratification
	of $\mathcal H_n$ (see \S {\ref{sec-stratification}}) \cite{Kato, Rellich}.  
	If $\lambda_k(h)$ is simple then $\lambda_k$ is analytic near $h$ and
	one may search for its critical points on $\TT_h$. The torus $\TT_h$ and its quotient $\mathcal M_h$
	are preserved under complex conjugation, and the function $\lambda_k$ is invariant under 
	complex conjugation.   The simplest critical points of $\lambda_k$ are the {\em symmetry points} 
	(\S {\ref{subsec-symmetry}}): 	the points $h' \in \TT_h$ (or $[h'] \in \mathcal M_h$)
	fixed by complex conjugation,  i.e. the real symmetric matrices in $\TT_h$.

 The set of symmetry points of $\TT_h$ is denoted $ \mathcal{S}(h) $.  
 If $h$ is real symmetric then $\mathcal S(h)$ consists precisely of the various signings of $h$.  
  Following \cite{BanBerWey}, we show: 
 \paragraph{Theorem \ref{thm-index}} 
 Each critical point 
	$h'\in \TT_h$ with simple eigenvalue $\lambda_k(h')$ and nowhere vanishing eigenvector is necessarily in the gauge equivalence class of a symmetry point. In other words, its image
	$[h'] \in \mathcal M_h$ is a symmetry point. 
 Suppose that for each $k$ ($1 \le k \le n$) each critical point $h_{\alpha}\in\TT_{h}$ of $\lambda_k$ has $\lambda_k(h_{\alpha})$ as a simple eigenvalue with nowhere vanishing eigenvector.
	Then the average nodal count distribution is a binomial distribution\footnote{see
 Example (\ref{example-binomial})} with mean $ \beta/2 $ and variance $ \beta/4 $.   
	Consequently, if the average nodal distribution is not binomial then there must 
	exist critical points (of some eigenvalue) that are not symmetry points.  
	
	We give a homological
	characterization of symmetry points:

\paragraph{Theorem {\ref{prop-integral}}} Let $h \in \mathcal S(G)$ and $\alpha \in \mathcal A(G)$
which we may identify as a 1-form on $G$.  Then $h_{\alpha}$ is gauge equivalent to a symmetry point
if and only if $\int_{\xi}\alpha \equiv 0 \mod \pi$ for all cycles $\xi$, i.e. chains  $\xi \in C_1(G,\ZZ)$ with $\partial \xi = 0$.

\subsection{Classification of critical gauge-equivalence classes}\label{subsec-exceptional}
	In general, the nodal surplus distribution $P(\mathcal{S}(h))$ depends on  Morse data from all critical points of $\lambda_k$ (for all $k$), whether or not they are symmetry points. 	Following Theorem
	{\ref{thm-index}}, there are two possible types of non-symmetry critical points $[h'] \in \mathcal M_h$ of $\lambda_k$:
	\begin{enumerate}
		\item {\em exceptional critical points}, for which $\lambda_k(h')$ is simple but its eigenvector vanishes on one or more vertices.  In this case $ [h'] $ is (usually) a degenerate critical point 
  (see Theorem {\ref{thm-linkage}}):  it is contained in a larger critical submanifold. 
		\item {\em incorrigible critical points}, for which the multiplicity
		of $\lambda_k(h')$ is greater than one. In this case, $\lambda_k$ fails to be smooth and one must replace the usual Morse theory with stratified Morse theory (\cite{GM}).
	\end{enumerate}  
	Concerning the first case, suppose the eigenvector $v$ vanishes only at a single vertex $v_0$ 
	of the graph $G$. Suppose that $v_0$ has degree $\deg(v_{0})$. 
	
	\paragraph*{Theorem {\ref{thm-linkage}}} Assuming the critical point $[h'] \in \mathcal M_h$ is sufficiently
 generic\footnote{Specific conditions on $h'$ are given in Theorem {\ref{thm-linkage}} in \S \ref{subsec-critical manifold} } then it lies in a 
 nondegenerate (Morse-Bott) critical submanifold of $\mathcal M_h$, of dimension $\deg(v_{0})-3$, which 
	is diffeomorphic to the	configuration space of a particular planar linkage.  Its Morse index
	may be expressed in terms of spectral data.
	
	The configuration spaces of planar linkages are fascinating objects.  They have been extensively 
	studied and their homology is completely known,  cf.~\cite{Farber,Millson,Walker}.

	\medskip\noindent
	For the second case, when the multiplicity of $\lambda_k(h')$ is greater than one,
	G. Berkolaiko and I. Zelenko \cite{Berkolaiko2} have determined the {\em normal Morse data} for $\lambda_k$,
	and its Betti numbers, which forms the central ingredient required
	for stratified Morse theory.  However, in order to apply stratified Morse theory to the mapping $\lambda_k:\TT_h \to \RR$
	it is required that the manifold $\TT_h\subset \mathcal H_n$ should be Whitney stratified.  Its stratification comes by intersecting with the natural
	stratification of $\mathcal H_n$ (cf.~\S {\ref{sec-stratification}}), but this requires that $\TT_h$  
	should be transverse to the strata of the stratification of $\mathcal H_n$.  The challenge is to guarantee transversality of 
	the torus $\TT_h$ by a generic choice of the single element $h \in \mathcal S(G)$.
	The transversality lemma in \cite{Berkolaiko2} does not
	address this situation.  The first nontrivial case concerns the stratum $S_2(k)$
	where $\lambda_k$ has multiplicity 2.  Suppose $h_{\alpha}$ is a critical point of $\lambda_k$,
	an eigenvalue of multiplicity 2.  In \S \ref{subsec-graph-conditions} we define the notion of a 
	splitting of the graph $G$ by the eigenspace of $\lambda_k$.  	(A related condition
	was considered by L. Lov\`asz in \cite[\S 10.5.2]{Lovasz}.)
	
	\paragraph{Theorem \ref{prop-transverse}}  If the eigenspace of $\lambda_k(h_{\alpha})$ does not split 
the graph $G$ then the space $\mathcal H(G)$ is transverse to $S_2(k)$ at given point $h_{\alpha}\in \TT_h$.

\paragraph{Corollary \ref{cor-splitting}}  As above, if the eigenspace of $\lambda_k(h_{\alpha})$ does
not split $G$ then for generic choice $h' \in \mathcal S(G)$ the torus $\TT_{h'}$ is transverse to
the stratum $S_2(k)$ near $h_{\alpha}$.

\begin{ack} The authors would like to thank Gregory Berkolaiko for enlightening discussions and for his comments on earlier versions of this paper. The first author would like to thank Nikhil Srivastava and Theo McKenzie for useful discussions. The idea of Theorem \ref{thm-index} (\ref{item-average}) originated in joint discussions with Ram Band and Gregory Berkolaiko regarding metric graphs. The authors are very grateful to an anonymous referee for 
many thoughtful comments, which have considerably improved the paper.  
 The first author would like to thank the Institute for Advanced Study, as this work began when he was a member there. \end{ack}
\begin{funding}
The first author was supported by the Ambrose Monell Foundation and the Simons Foundation Grant 601948, DJ.
\end{funding}

	\section{Notation and definitions}\label{subsec-notation}  
	\subsection{Symmetric and Hermitian forms}\label{subsec-Hermitian}
	Let $\mathcal S_n$ denote the vector space of $n \times n$ real symmetric matrices, $\mathcal A_n$ the space
	of $n \times n$ real antisymmetric matrices and $\mathcal H_n$ the space of $n \times n$ Hermitian matrices, that is, matrices of linear operators on $\CC^n$ expressed in the standard basis and that are self-adjoint with
	respect to the standard Hermitian form $\langle x, y \rangle = \sum\bar{x_i}y_i$. 
	
	If $V \subset \CC^n$ is a complex subspace then the standard Hermitian form restricts to a 
	Hermitian form on $V$ and we denote by $\mathcal H(V)$ the self adjoint linear operators 
	$V \to V$.  If $\xi \in \mathcal H_n$ then it may fail to preserve $V$ however its ``restriction" to $V$
	may be defined by expressing $\xi = \left( \begin{smallmatrix} A & B \\ B* & D \end{smallmatrix} \right)$ with respect to the decomposition
	$\CC^n = V \oplus V^{\perp}$.  The restriction $\xi|V$ is defined to be the  operator $A\in \mathcal H(V)$.  Equvalently, $\xi|V$ is the operator corresponding to the
	restriction to $x, y \in V$ of the {\em sesquilinear} form $(x,y)_{\xi} = \langle x, \xi y \rangle$.

	\subsection{Laplace and Schr\"odinger operators}
	\label{subsec-support}  Throughout this section we fix a graph $ G=G([n],E) $.
	The natural ordering on the set of vertices $ [n]:=\{1,2,\ldots,n\} $ determines an orientation for each edge.
	Write $r \sim s$ if $r \ne s$ and vertices $r, s$
	are joined by an edge.  Write $r \simeq s$ if $r \sim s$ or $r=s$.

	A (real or complex) {\em  matrix supported on $G$} is an $n \times n$  matrix $h$ such that $h_{rs} \ne 0 \implies r \simeq s$. Such a matrix is
 {\em properly supported} on $G$ if, in addition, $r \sim s \implies h_{rs} \ne 0$.
	Symmetric, antisymmetric and Hermitian matrices supported on $G$ are denoted $\mathcal S(G), \mathcal A(G), \mathcal H(G)$ respectively.
	Examples of matrices in $\mathcal S(G)$ include the
	{\em adjacency matrix} for $G$, (weighted) {\em Laplace operators} for $G$ and discrete Schr\"odinger operators,
	see \S \ref{subsec-intro1} above. More generally, any matrix $h \in \mathcal H(G)$  may be
 considered a {\em magnetic  Schr\"odinger operator} for $G$, 
	(see \S {\ref{subsec-star-action}} below and references \cite{CdeV1,CdeV2}).
	
	\subsection{Graph homology}  
	The space $C_0(G;\ZZ)\cong \ZZ^n$ of 0-chains is the vector space of formal linear combinations of vertices, $\sum_{r=1}^n c_r[r]$.
	Each edge $rs$ with $r<s$ is orientated from $r$ to $s$ so that the group $C_1(G;\ZZ)$ of $1$-chains is the
	group of formal linear combinations 
	\begin{equation}\label{eqn-xi}
		\xi = \sum_{\substack{{r \sim s}\\{r<s}}}\xi_{rs}[rs] \ \text{ and }\ \xi_{rs} \in \ZZ.\end{equation}
	Then $H_1(G;\ZZ) = \ker(\partial)$ where $\partial:C_1(G;\ZZ) \to C_0(G;\ZZ)$ with $\partial[rs] = [s]-[r]$.
	The first Betti number is
	\[\beta = {\mathrm {rank}} H_1(G,\ZZ) =|E|-n+c\] 
	where $c$ is the number of connected components of $G$.
	
	The vector space $\RR^n$  may be viewed as the space of real-valued functions $\Omega^0(G)$ on the vertices of $G$.  If
	$v = (v_1,v_2,\cdots, v_n)$ we sometimes write $v_r = v(r)$. 
	The vector space $\mathcal A(G)$ of real, antisymmetric matrices supported on $G$ may be viewed as the
	space of $1$-forms $\Omega^1(G)$ on $G$ with coboundary differential
	\begin{equation}\label{eqn-differential}
		d: \Omega^0(G) = \RR^n \to \Omega^1(G) = \mathcal A(G);\quad (df)_{rs} = 
		\begin{cases} f(s) - f(r) &\text{if } r \sim s\\ 0 &\text{otherwise.}\end{cases}\end{equation}
	There are no 2-forms on a graph so $H^1(G;\RR) = \Omega^1(G)/d\Omega^0(G)$ is
	canonically dual to the homology $H_1(G;\RR)$ under the
	the natural pairing that is determined by integration $\Omega^1(G) \times C_1(G;\RR)\to \RR$.  If $\alpha\in\mathcal A(G)=\Omega^1(G)$ and $\xi\in C_1(G;\RR)$ as in (\ref{eqn-xi}), then
	\[ \int_{\xi}\alpha  = \sum_{\substack{{r \sim s}\\{r<s}}} {\xi_{rs}\alpha_{rs}}\]

	\subsection{Action of \texorpdfstring{$\mathcal A_n$}{An}}\label{subsec-star-action}
	The vector space $\mathcal A_n=\mathcal A_n(\RR)$ of $n \times n$ real antisymmetric matrices acts on the vector space
	$\mathcal H_n$ of $n \times n$ Hermitian matrices by
	\[ (\alpha*h)_{rs} = e^{i\alpha_{rs}}h_{rs}\]
	for all $\alpha \in \mathcal A_n(\RR)$ and $h \in \mathcal H_n$ with $(x+y)*h = x*(y*h)$ and with $0*h = h$. 
	Then $\mathcal A(G)$ acts on $\mathcal H(G)$.

	If $h$ is the discrete Schr\"odinger operator then $\alpha *h$ may be
	interpreted as the corresponding magnetic Schro\"odinger operator in the presence of a {\em magnetic field} 
	described by $\alpha$, whose flux through a cycle $\xi$ is $\int_{\xi}\alpha$, with a sesquilinear form 
	\begin{equation}\label{sesquilinear form}
		\langle f, (\alpha*h)f \rangle = -\sum_{r \sim s} h_{rs}\left|f(s) - e^{i\alpha_{rs}}f(r)\right|^2 + \sum_{r=1}^n V(r) \left|f(r)\right|^2,  
	\end{equation}
	instead of the quadratic form of $h$ in \eqref{quadratic form}.
	If $h = (h_{rs}) \in \mathcal H_n$ define $|h| \in \mathcal S_n$ by $|h|_{rs} = |h_{rs}|$ for $r \ne s$ and $|h|_{rr} = h_{rr}$ (diagonal entries of $|h|$ can be negative).
	Then there exists $\alpha \in \mathcal A_n$ so that $h = \alpha*(|h|)$.   
	
	\subsection{Gauge invariance}\label{subsec-gauge}
	The $*$ action factors through the torus $ \mathcal A_n(\RR)/\mathcal A_n(2 \pi \ZZ)$.
	So the subtorus {\em supported on $G$} 
	\[  
	\TT(G):=  \left\{\alpha\in\mathcal A_n(\RR)/\mathcal A_n(2 \pi \ZZ)\left| \ \alpha_{rs}\ne 0\right. \implies r\sim s\right\},
	\]
	acts on $\mathcal H(G)$ by the $*$ action.
	The differential (\ref{eqn-differential}) also factors 
 \begin{equation}\label{eqn-diag-d}
 \begin{diagram}[size=2em]
\Omega^0(G) = \RR^n & \rTo^{d} & \Omega^1(G) = \mathcal A(G)\\
\dTo^{\mod {2\pi}} && \dTo_{\mod{2\pi}} \\
\TT^n  & \rTo^{d} & \TT(G)
 \end{diagram}\end{equation}
through the {\em gauge group} $\TT^n = \RR^n/(2\pi\ZZ)^n$.
	{\em Gauge invariance} is the statement that the * action by coboundaries is simply given by conjugation:
	for any $\theta \in \TT^{n}$ and any $h \in \mathcal H_n$, direct calculation gives
	\begin{equation}\label{eqn-gauge-transformation} d\theta *h = e^{i \theta} h e^{-i\theta}\end{equation}
	where  $\theta=(\theta_1,\theta_2,\cdots,\theta_n) \in \TT^n$ and $ e^{i\theta} = {\mathrm {diag}}(e^{i\theta_1}, e^{i\theta_2}, \cdots, e^{i\theta_n})$.
	The * action by $d\theta$ preserves eigenvalues and preserves eigenvectors {\em up to phase}:  if $E_{\lambda}(h) = \ker(h - \lambda I)$ then
	\begin{equation}\label{eqn-phase}  E_{\lambda}(d\theta *h) = e^{i\theta}E_{\lambda}(h).
	\end{equation}
 Elements $h, h' \in \mathcal H(G)$ that differ by a gauge transformation 
 ($h' = d\theta *h$) are said to be {\em gauge equivalent}.
Gague equivalence determines an identification (cf. (\ref{eqn-figT(G)}) of the 
 {\em quotient torus}  (the manifold of magnetic fields modulo gauge transformations) with cohomology:
	\[\TT^{\mathcal A/d}(G) := \TT(G)/d(\TT^n) \cong H^1(G;\RR/2\pi \ZZ).\]
	
	\subsection{The embedded torus and its symmetry points}\label{subsec-symmetry}
	Recall (\ref{subsec-support}) that a matrix $h \in \mathcal H(G)$ is {\em properly supported} on $G$ if $h_{rs} \ne 0$ whenever $r \sim s$. (Diagonal entries
 $h_{rr}$ may vanish.) Such $h $ defines a mapping $\TT(G) \to \mathcal H_n$ by $\alpha \mapsto \alpha *h$, whose image is an embedding of $\TT(G)$ into $\mathcal H(G)$, 
	\begin{equation}\label{eqn-TAh}
		\TT_h := \TT(G)*h = \left\{ \alpha * h \ :\ \alpha \in \TT(G) \right\}
		= \left\{ \alpha * |h|\ :\ \alpha \in \TT(G) \right\}.\end{equation}
	We refer to $\TT_h$ as {\em the embedded torus}. For $h \in \mathcal H(G)$ which is not properly supported on $G$, the dimension of the embedded torus $\TT_h$ is the number of nonzero elements $h_{rs}$ with $r<s$. The embedded torus is invariant 
	under complex conjugation and we refer to the set of its fixed points (i.e. the real points)
	\[\mathcal{S}(h)\ :=\ \TT_h\cap\mathcal{S}(G)\ =\ \{\alpha*|h|\ :\  \alpha\equiv 0 \mod \pi\}.\]
	as {\em symmetry points}.  If $h \in \mathcal S_n$ then its symmetry points $\mathcal{S}(h)$ consist of symmetric matrices $h'$ 
	obtained from $h$ by changing the signs in any subset of off-diagonal entries $h_{rs}$ or equivalently
	\begin{equation}  \label{eqn-mod-pi} h' = \alpha *h \text{ where } \alpha \equiv 0 \mod \pi.\end{equation}
	
	The action of the {\em integral gauge group} $(\pi \ZZ)^n \subset \RR^n$
	preserves the set of symmetry points and changes the signs of the components of the corresponding eigenvectors.
	The set $\mathcal S(h)$ decomposes into a
	union of orbits under the integral gauge group. If $h$ is properly supported on $G$ ($h_{rs} \ne 0$ whenever $r \sim s$) then $\mathcal S(h)$ has $2^{|E|}$ elements,
	partitioned into $2^{\beta}$ orbits (cf.~\S \ref{subsec-magnetic-perturbations}). Each orbit corresponds to a
choice of parity of the circulations around a choice of elementary cycles. 
	
	\begin{thm}\label{prop-integral}
		Suppose $h\in \mathcal H(G)$ is properly supported on $G$.
  Let $\alpha \in \mathcal A(G)=\Omega^1(G)$ 
so that $h = \alpha * |h|$. Then $h$ is gauge-equivalent to a symmetry point $h'\in\mathcal S(h)$ if and only if
		\begin{equation}\label{eqn-modpi} \int_{\xi}\alpha \equiv 0 \mod \pi\end{equation}
		for all cycles $\xi$, i.e. chains  $\xi \in C_1(G,\ZZ)$ with $\partial \xi = 0$.
	\end{thm}
	
	\begin{proof} Since $h$ is properly supported, the element $\alpha$ is uniquely determined modulo $2\pi\ZZ$.
 If $h$ is a symmetry point then $\alpha \equiv 0 \mod \pi$ so (\ref{eqn-modpi}) holds.
		If $h$ changes by gauge-equivalence, the integral (\ref{eqn-modpi}) is unchanged, by Stokes' theorem.

On the other hand, if (\ref{eqn-modpi}) holds for all cycles then by duality the cohomology class $[\alpha] $
		vanishes in $H^1(G;\RR)/H^1(G;\pi\ZZ)$ so it lies in  $H^1(G; \pi \ZZ)\subset H^1(G;\RR)$ and it comes from a 1-form
		$\alpha' \in \Omega^1(\pi\ZZ)$, that is, an antisymmetric matrix whose entries are multiples of $\pi$.  
	Then the cohomology class $[\alpha' - \alpha] \in H^1(G;\RR)$ vanishes so there exists $\theta \in \Omega^0(G;\RR)$ with
$\alpha' = \alpha + d\theta$.  This proves that
the symmetry point $\alpha'*|h|$ is gauge-equivalent to $h=\alpha*|h|$
	\end{proof}

	\subsection{Eigenvalues as Morse functions} \label{subsec-magnetic-perturbations}
	Eigenvalues of elements $h \in \H_n$ are real and ordered, say
	\[ \lambda_1(h) \le \lambda_2(h) \le \cdots \le \lambda_n(h).\]
	For each $k$ ($1 \le k \le n$) the mapping $\lambda_k:\H_n \to \RR$ is well defined,
	continuous and piecewise real-analytic:  there is a stratification 
  of $\H_n$ by analytic subvarieties such that
	the restriction of $\lambda_k$ to each stratum is analytic  
 (cf. \S \ref{subsec-SmRm} and Lemma \ref{lem-Smk}).
 
	The restriction of each $\lambda_k$ to the embedded torus $\TT_{h}$ is invariant under gauge transformations so it determines a function on the quotient,
	\begin{equation}\label{eqn-magnetic-manifold}
		\mathcal M_h = \TT_{h}\sslash\TT^n,\end{equation} 
   where we use the notation $\sslash \TT^n$ to denote dividing by gauge equivalence.
  The torus $\mathcal M_h$ has dimension $\beta$, and
	is referred to in \cite{CdeV2} as the {\em manifold of magnetic perturbations modulo gauge transformations}.

 \begin{equation}\label{eqn-figT(G)}
	\begin{diagram}[size=2.5em]
		\TT(G) & \rTo^{*h} & \TT_h & \rInto & \mathcal H_n & \rTo^{\lambda_k} & \RR&\\
		\dTo^{d(\TT^n)} && \dTo^{\sslash \TT^n} & && \ruTo(4,2)_{\lambda_k} &&\\
		\TT^{\mathcal A/d}(G) & \rTo & \mathcal M_h &&&&&&
	\end{diagram}\end{equation}
 If $\alpha \in \TT(G)$ then the equivalence class of $\alpha *h$ in $\mathcal M_h$ is denoted $[\alpha *h]$ or $[h_{\alpha}]$.
	
	If $\theta \in \RR^n$ then $\overline{(d\theta)*h} = d(-\theta)*\bar h$ so complex conjugation passes to an involution on $\mathcal M_h$.
	{\em Every fixed point of this involution comes from a symmetry point in $\TT_h$}: for if $h \in \mathcal H_n$ and $[h]\in \mathcal M_h$ is fixed, this means
	$\bar h = (d\theta) *h$ for some $\theta \in \RR^n$, so $(d\frac{\theta}{2})*h$ is a symmetry point.  
	It is therefore reasonable to refer to these fixed points of $\mathcal M_h$ as {\em symmetry points of} $\mathcal M_h$.
	
	\begin{lem}\label{lem-symmetry-count} Let $G$ be a simple graph with $c$ connected components and let $h\in\mathcal H_{n}(G)$, properly supported on $G$. 
		Then, each symmetry point $h' \in \mathcal S(h)$ has exactly $2^{n-c}$ gauge-equivalent symmetry points. 
		Thus, the number of symmetry points in $\mathcal M_h$ is $2^{|E|-(n-c)}=2^{\beta}$.
	\end{lem}
	\begin{proof}
		It is enough to consider the case of real symmetric $h \in \mathcal S_n$, in which case its gauge-equivalent symmetry points are
		\[ [h]\cap \mathcal S(h) = \left\{ df * h \big| \ f(r)\in \{0,\pi\} \text{ for all } r \right\}. \]
		There are $2^n$ choices for $f$ among which $2^{c}$ are in the kernel of $d$ (those which are constant on connected components of $G$).
  So there are $2^{n-c}$ distinct values for $df$, and therefore $2^{n-c}$ distinct values of $df*h$ since $h$ is properly supported on $G$.  
		Hence, $[h]$ contains exactly $2^{n-c}$ gauge-equivalent symmetry points. Repeating this argument for any other $h'\in S(h)$ leaves $2^{|E|-(n-c)} = 2^{\beta}$ equivalence classes of symmetry points in $\mathcal M_h$.  
	\end{proof}

	\subsection{Nodal surplus}\label{subsec-nodal-surplus}
	Generalizing the notions described in the introduction, let $h$ be a Hermitian matrix supported on $G$, suppose
	$\lambda_k(h)$ is a simple eigenvalue with nowhere vanishing eigenvector $v = (v_1,v_2,\cdots,v_n)$.  
	Further assume that $\bar{v}_rh_{rs}v_s\in\RR$ for all $r\sim s$ (which is equivalent to $h$ being a critical point of $\lambda_k$, see 
	Theorem \ref{thm-index} part (\ref{item-more-generally})).    
	Define the {\em nodal count} $\phi(h,k)$  to be the number of edges $r\sim s$ such that
	\begin{equation}\label{eqn-nodal-count} \bar{v}_rh_{rs}v_s> 0.\end{equation}
	The {\em nodal surplus} is the number $\phi(h,k)-(k-1)$.  This number does not change under gauge transformation and
	it is known (see Theorem \ref{thm-index} below) that the nodal surplus is between $0$ and $\beta$, the first Betti number of $G$.
	The {\em nodal surplus distribution} $P(h)=\left(P(h)_{0},P(h)_{1},\ldots,P(h)_{\beta}\right)$ is the vector representing the probability distribution of these numbers over the $n$ possible eigenvalues:
	\[P(h)_s = \frac{1}{n} \# \big\{1 \le k \le n\big|\ \phi(h,k)-(k-1) = s\big\}.
	\]
	Assuming that $h \in \mathcal{S}_n$ and all its signings $h' \in \mathcal S(h)$ have all eigenvalues simple with nowhere-vanishing eigenvectors, the
	distribution can be averaged over signings to give the {\em average nodal distribution}
	\[P(\mathcal S(h)) = 2^{-|E|}\sum_{h' \in \mathcal S(h)}P(h').\]

	\section{Morse theory}
	\subsection{Critical points}
	Throughout this section we fix a graph $G$ with vertices $ 1, \cdots, n$ and edges $r \sim s$.
	Let $h \in \mathcal S(G)$ be a real symmetric matrix properly supported on $G$, cf.~ \S \ref{subsec-support}.
	For $\alpha \in \mathcal A(G)$ denote by $h_{\alpha} = \alpha*h$ the magnetic perturbation of $h$.
	Fix $k$ and write $\lambda_k(\alpha) = \lambda_k(h_{\alpha})$ for the k-th eigenvalue.  
	Let $\mathcal M_h$ be the manifold (\ref{eqn-magnetic-manifold}) of magnetic perturbations of $h$ modulo gauge transformations. It is a torus of dimension $\beta$, the first
	Betti number of the graphs $G$. By equation (\ref{eqn-phase}) the eigenvalue $\lambda_k(\alpha)$ 
 of an element $[h_{\alpha}] \in \mathcal M_h$, and its multiplicity are well defined; 
 and whether or not an eigenvector vanishes at a given vertex is well defined.  
	
	We consider $\lambda_k:\mathcal M_h \to \RR$
	to be a sort of generalized Morse function. If $\lambda_k$ is smooth at a point $x=[h_{\alpha}] \in \mathcal M_h$ (in which case it is also analytic) we say that
	$x$ is a {\em smooth point} of $\lambda_k$. A {\em critical point} of $\lambda_k$ is either a non-smooth point
	or a smooth point where $\nabla\lambda_k(x) = 0$. Consider the following possibilities:
	\begin{enumerate}
		\item[(0)] $x$ may be a \textbf{smooth, regular} (i.e., not critical) point of $\lambda_k$.
		\item[(1)]  $x$ may be a \textbf{symmetry} point of $\mathcal M_h$.
		\item[(2)] $x$ may be a \textbf{non-symmetry, smooth}, (possibly degenerate) \textbf{critical} point of $\lambda_k$ .  
		\item[(3)] $x$ may be a \textbf{non-smooth} point of $\lambda_k$.
	\end{enumerate}

	\begin{thm} \label{thm-index}
		Fix properly supported $h \in \mathcal S(G)$. Consider $\lambda_k:\mathcal M_h \to \RR$ as above. 
		\begin{enumerate}
			\item \label{item-symmetry}
			Every symmetry point of $\mathcal M_h$ is a  critical point of $\lambda_k$.
			\item \label{item-perfect}
			If the only critical points of $\lambda_k$ on $\mathcal M_h$ are the symmetry points and if they are nondegenerate
			then the number of such critical points of index $s$ is $\left( \begin{smallmatrix} \beta \\ s \end{smallmatrix} \right)$.
			\item\label{item-more-generally} Suppose $h_{\alpha} \in \TT_{h}$ has a  simple eigenvalue $\lambda_k(h_{\alpha})$ with eigenvector $v$. Then
			$(h_{\alpha})_{rs}\bar v_r v_s$ is real for all $ r \sim s$ if and only if $h_{\alpha}$ is a critical point of $\lambda_{k}$ as a function on $ \TT_{h} $, in which case 
			$h_{\alpha}$ is gauge equivalent to a matrix $h'$ such that $h'_{rs} \notin \RR \implies \bar{v}_rv_s = 0$.
			\item\label{item-nonvanishing}
			In particular, if $h_{\alpha}\in\TT_{h}$ is a critical point of $\lambda_k$ and $\lambda_k(h_{\alpha})$ is {\em simple} with nowhere vanishing eigenvector  then $[h_{\alpha}] \in \mathcal M_h$ is a symmetry point.\\
			(Equivalently, there exists $ \theta\in\TT^{n} $ such that $h_{\alpha+d\theta}\in \mathcal S(h)$.)
			\item \label{item-nodal-Morse} A critical point $[h_{\alpha}]\in \mathcal M_h$ as in (4) is nondegenerate and its Morse index is the nodal surplus, $\phi(h_{\alpha},k)-(k-1)$.
			\item\label{item-equal-diagonals} If the diagonal entries of $h$ are all equal then the average nodal count distribution is symmetric,
			\[P(\mathcal S(h))_s =P(\mathcal S(h))_{\beta-s},\qquad s\in\{0,1,\ldots,\beta\}.\]
			
			\item \label{item-average} Suppose that for each $k$ ($1 \le k \le n$) each critical point $h_{\alpha}\in\TT_{h}$ of $\lambda_k$ has $\lambda_k(h_{\alpha})$ as a simple eigenvalue with nowhere vanishing eigenvector.
			Then the average nodal count distribution is binomial:
			\[ P(\mathcal S(h))_s = 2^{-\beta}\left(\begin{matrix} \beta \\ s \end{matrix} \right).\]
			
		\end{enumerate}
		
	\end{thm}
	Parts (\ref{item-symmetry}) and (\ref{item-nodal-Morse}) of Theorem \ref{thm-index} are due to Berkolaiko and Colin de
	Verdi\`ere\footnote{In both \cite{Berkolaiko} and \cite{CdeV2} and  the matrix $h$ was assumed to be real symmetric but essentially the same proof works in
		general.} \cite{Berkolaiko,CdeV2}. Part (\ref{item-perfect}) is an immediate consequence,
	also known to both of these authors. Part (\ref{item-more-generally}) was already observed in \cite[Theorem A.1 and Lemma A.2 ]{BanBerWey}. Part (\ref{item-nonvanishing}) is an immediate consequence known to the authors of \cite{BanBerWey}. It says that the only simple critical points of $\lambda_k| \TT_h$ with non-vanishing eigenvector occur along the 
	intersection of $\TT_h$ with the conjugacy classes of symmetry points: the $ 2^{|E|} $ real elements $ \mathcal{S}(h) $.
	Proofs for Theorems \ref{thm-index} and \ref{thm-linkage} below will appear in \S  \ref{sec-proofs1} and \S \ref{sec-proofs2}.
 
\subsection{Example of matrices with Binomial nodal count distribution}\label{example-binomial}

Let $h = h_0 + \eta V$ be a Schr\"odinger operator on the complete graph, with $h_0$ properly supported (i.e., $(h_{0})_{rs}\ne0$ for all $r\ne s$), $V=\diag(V_{1},\ldots,V_{n})$
 with distinct entries, and $\eta \in \RR$. If $\eta$ is sufficiently large, then all matrices $\alpha*h \in \TT_h$ will have simple eigenvalues and nowhere vanishing eigenvectors, so
 $P(S(h))$ is binomial.

 To see that, set $\epsilon=\frac{1}{\eta}$ and let $h_{\epsilon}= \epsilon h= V+\epsilon h_{0}$. We treat $\alpha*h_{\epsilon}=V+\epsilon(\alpha*h_{0})$ as a small perturbation of $V$ whose distinct eigenvalues are $V_{j}$ with eigenvectors $e_{j}$ for $j=1,\ldots,n$. The min-max principle gives $|\lambda_{j}(\alpha*h_{\epsilon})-V_{j}|\le \max_{rs}|\epsilon(\alpha*h_{0})_{rs}|=\max_{rs}|\epsilon(h_0)_{rs}|$ so there is a uniform constant $C>0$ such that when $0<\epsilon<C$, the eigenvalues of $\alpha*h_{\epsilon}$ are distinct, for every $\alpha$. Suppose $0<\epsilon<C$ and let $v_{\epsilon}$ be the $j$-th eigenvector of $\alpha*h_{\epsilon}$. Comparing $v_{\epsilon}$ to $e_{j}$, perturbation theory gives $v_{\epsilon}(j)=1+O(\epsilon^{2})\ne 0$, and for $i\ne j$,
 \[|v_{\epsilon}(i)|=\epsilon\left|\frac{(\alpha*h_{0})_{ij}}{V_{i}-V_{j}}\right|+O(\epsilon^{2})\ge\epsilon\min_{r<s}\left|\frac{(h_{0})_{rs}}{V_{r}-V_{s}}\right|+O(\epsilon^{2})\ne0, \]
for sufficiently small $\epsilon$, uniformly in $\alpha$.



	\section{Exceptional critical points and the linkage equation}\label{subsec-linkage}
	In this section we consider the case where $ [h_{\alpha}]\in\mathcal{M}_{h} $ is an exceptional critical point of 
 $ \lambda_{k} $ (cf.~\S \ref{subsec-exceptional}). That is, $ [h_{\alpha}] $ is a non-symmetry, smooth, critical point with 
 $ \lambda_{k}(h_{\alpha}) $ simple. According to Theorem \ref{thm-index} the eigenvector $ v $ corresponding to $ \lambda_{k}(h_{\alpha}) $ 
 vanishes somewhere. (By generic choice of $ h $ we can guarantee that every eigenvector of $h$ is nowhere vanishing (cf.~\cite{Urschel}) but we cannot guarantee the same holds for all $h_{\alpha}\in \TT_h$.) We address the simple case of eigenvector 
 $ v $ that vanishes at a single vertex. By possibly replacing $ h_{\alpha} $ with a gauge equivalent $ h_{\alpha+d\theta} $ and $ v $ with 
 $ e^{i\theta}v $ we may assume that $ v $ is real with non-negative entries. 
 The setting for Theorem \ref{thm-linkage} is described next.    
	\subsection{The Setting}\label{subsec-setting}
	To simplify the notation we assume the graph $G$ has $n+1$ vertices labeled  $0,1,2,\cdots,n$, with corresponding properly supported real symmetric matrix $h \in \mathcal{S}(G)$. Suppose $h_{\alpha} = \alpha*h$ is a critical point of $\lambda_k$ with a simple eigenvalue $ \lambda:=\lambda_{k}(h_{\alpha}) $ and a normalized eigenvector $v = (v_0,v_1,\cdots, v_n)=(0,v')$ and $v_0 = 0$, $v_r> 0$ for $1 \le r \le n$. Writing $ h $ and $ h_{\alpha} $ as block matrices in the $ \RR^{n+1}=\RR\oplus\RR^{n} $ decomposition gives
	\begin{equation}\label{eqn-block-matrix} 
		h = \left(\begin{matrix} a & b \\ b^* & D\end{matrix} \right) \ \text{ and }\ h_{\alpha} = \left(\begin{matrix}a & b_{\alpha} \\ b_{\alpha}^* & D_{\alpha} \end{matrix} \right)
		\text{ with } \left( \begin{matrix} a & b_{\alpha}\\
			b^*_{\alpha} & D_{\alpha} \end{matrix} \right) \left(\begin{matrix} 0 \\ v' \end{matrix}\right)
		= \left( \begin{matrix} 0 \\ \lambda v' \end{matrix} \right). \end{equation}
	Let $ E_{0} $ be the edges connected to vertex $ 0 $. For convenience, write $ r\in E_{0} $ if $ 0r\in E_{0} $. Let $ H $ be the induced subgraph of $ G $ on the vertices $ r\ge 1 $. Thus, $ H $ is obtained from $ G $  by removing vertex $ 0 $ and its edges $ E_{0} $. Then, $ a\in\RR,\ b\in\RR^{E_{0}},\ b_{\alpha}\in\CC^{E_{0}},\  D\in\mathcal{S}(H) $, and $ D_{\alpha}\in\mathcal{H}(H) $. In fact, since $ h_{\alpha} $ is critical and $ v_{r} $ is real and non-zero for $ r\ge 1 $, then $ D_{\alpha} $ is real by part (\ref{item-more-generally}) of Theorem \ref{thm-index}. Hence, $ D_{\alpha}\in\mathcal{S}(H)$ is a signing of $ D $. 

	The vector $ b_{\alpha} $ has the form $ (b_{\alpha})_{r}=e^{i\alpha_{0r}}b_{r} $ for $ r\in E_{0}$. Let $ M_{r}:=|b_{r}v_{r}|>0$ and $ \theta_{r}\in\RR/2\pi\ZZ  $ be the polar coordinates of $ (b_{\alpha})_{r}v_{r}=M_{r}e^{i\theta_{r}} $ for every $ r\in E_{0} $.

		\subsection{Configuration space of a planar linkage}\label{subsec-planar}
		Equation (\ref{eqn-block-matrix}) implies that the following
	{\em planar linkage equation} (\cite{Walker, Farber, Millson}) holds:
	\begin{equation}\label{eqn-linkage}
		b_{\alpha}.v' = \sum_{r \in E_0} e^{i \theta_r}M_r = 0.\end{equation}
		This equation (\ref{eqn-linkage}) describes a collection of vectors $M_re^{i\theta_r}\in \CC=\RR^2$ in the plane, placed end to tail, that starts
		and ends at the origin, that is, a planar linkage, depending on a collection of lengths $L = \{M_r\}_{r\in E_{0}}$.	Let $S^1\subset \CC$ be the unit circle. The {\em configuration space} $\Theta_L$ (see \cite{Farber}) of the planar linkage defined by (\ref{eqn-linkage}) is the set of solutions modulo rotations, that is, 
		\begin{equation}\label{eqn_configuration_space}
			\Theta_L =\left\{(e^{i\theta_{r}})_{r\in E_{0} } : \underset{r\in E_0}{\Sigma} 
        e^{i \theta_r}M_r = 0\right\}/S^{1}\ \subset\  (S^{1})^{E_0}/S^{1},
		\end{equation}
		where the unit circle acts diagonally on $(S^1)^{E_{0}}$ by multiplication. The planar linkage is said to be {\em generic} if for any $ \epsilon\in\{-1, 1\}^{E_{0}} $,
		\begin{equation}\label{eqn_generic_linkage}
			\sum_{r\in E_{0}} \epsilon_rM_r \ne 0.
		\end{equation}
	Let $ M_{s} $ be the maximal length, $ M_{s}=\max(M_{r})_{r\in E_{0}} $. If $ M_{s}>\sum_{r\ne s}M_{r}$ then there are no solutions, $ \Theta_{L}=\emptyset $. If the planar linkage is generic and 
 $ M_{s}<\sum_{r\ne s}M_{r}$,  
 then $ \Theta_{L} $ is a smooth manifold of dimension $|E_{0}|-3$ (\cite{Millson, Farber, Walker}) whose Betti numbers	have been computed in \cite{Hausmann, Farber}.
	Let $M_{t}$ be the second largest length.  If $M_s + M_t \le \frac{1}{2}\sum_r M_r$ then $\Theta_L$
	is connected, otherwise it has two connected components, exchanged by complex conjugation, each diffeomorphic to the torus of dimension  $|E_0|-3$.

	\subsection{Exceptional points}\label{subsec-critical manifold} In the notation of \S \ref{subsec-setting}, suppose $h_{\alpha}=\alpha*h$ is an exceptional critical point of $\lambda_k$ with real eigenvector $v = (0,v')$ and simple eigenvalue $\lambda = \lambda_k(h_{\alpha})$. 
 The complex conjugate point $\bar{h}_{\alpha}=(-\alpha)*h$ is also a critical point of $\lambda_k$, with the same eigenvalue $\lambda$ and eigenvector $ v$. Moreover, $\lambda$ is also an eigenvalue of $D_{\alpha}$, say
    $\lambda = \lambda_{k'}(D_{\alpha})$ is its $k'$-th eigenvalue.
    
	 Let $F$ be the connected component of the critical set in $\mathcal M_h$ of $\lambda_k$ that contains $h_{\alpha}$, union with the connected component of the critical set of $\lambda_k$ that contains $\bar h_{\alpha}$, noting that these two sets may be the same\footnote{Thus, the set $F\subset \mathcal M_h$ has either one or two connected components.}. 
\begin{thm} \label{thm-linkage}  Assume the following:
	 
\begin{enumerate}
	\item The eigenvalue $\lambda =\lambda_{k'}(D_{\alpha})$  is simple.
	\item The collection $\{M_r = |b_{r} v_r|\}_{r \in E_0}$  is generic \eqref{eqn_generic_linkage}.
    \item For any $[h']\in F$ the eigenvalue $\lambda = \lambda_k(h')$ is simple, and
\begin{equation}\label{eqn-ch}c(h')=\sum_{j\ne k}\frac{|\psi_{j}(0)|^{2}}{\lambda_{k}(h')-\lambda_{j}(h')}\ne0,\end{equation}
  where $ (\psi_{j})_{j=1}^{n+1} $ are a choice of orthonormal eigenvectors of $ h' $ corresponding to the ordered eigenvalues.
			\end{enumerate}
Then the critical set $F$ coincides with the explicitly defined set
\begin{equation}
    F'=\left\{[h']\in\mathcal M_h: h'v=\lambda v \text{ and there exists }\alpha_0'\in\TT(E_{0})
    \text{ such that } h'=\alpha_{0}'*h_{\alpha} \right\}.
                \end{equation}
It is a nondegenerate (Morse-Bott) critical submanifold  of dimension $ |E_{0}|-3 $ which is diffeomorphic to the configuration space $ \Theta_L $. Moreover, the Morse index of this critical submanifold is equal to 
			\[\ind(F)=\phi(D_{\alpha},k')-(k'-1)+\begin{cases}
				2 & \text{\emph{if}}\ c(h')<0\\
				0 & \text{\emph{if}}\ c(h')>0
			\end{cases}.\]  	
		\end{thm}
\begin{rems}  Recall that the pseudo-inverse $B^{+}$ of a Hermitian matrix $B$ with kernel $V$
has the same kernel $V$ and acts as $B|_{V^{\perp}}^{-1}$ on $V^{\perp}$.
If we define the resolvant $(h'-z)^{-1}$ at $z=\lambda_{k}(h')$ using the pseudo inverse $A=\left(h'-\lambda_{k}(h')\right)^{+}$, then $c(h')=A_{0,0}$. 

A related observation for periodic metric (quantum) graphs appears in
\cite{BerkolaikoKha} \S 3.4, where certain graphs are constructed so that the maximum of their first spectral band is obtained on a critical manifold which is a planar linkage configuration space. 
\end{rems}

	\section{Proof of Theorem \ref{thm-index}}\label{sec-proofs1}
	\subsection{} For part (\ref{item-symmetry}), suppose $h'\in\TT_{h}$ is a symmetry point (of $\TT_{h}$), namely $h' = \bar h' \in \mathcal S(h)$. If $h'$ is not a smooth point of $\lambda_{k}$, then it is a critical point. Suppose $h'$ is smooth, then 
	the directional derivative of $\lambda_k$ in the direction $\alpha \in \mathcal A(G)$ is $\frac{d}{dt}\lambda_k(t \alpha *h')\vert_{t=0}=0$  because
	\[ \lambda_k(t\alpha*h') = \lambda_k(\overline{t\alpha*h'}) = \lambda_k (-t\alpha*h').\]
	If $h''\in [h']$, then it is conjugate to $h'$. Conjugation takes a neighborhood of $h'$ in $\mathcal H(G)$ to a neighborhood of $h''$, 
	preserving the eigenvalue $\lambda_k$ so it also preserves the derivative of $\lambda_k$.

	Part (\ref{item-perfect}) follows immediately from the Morse inequalities, $C_i(\mathcal M_h) \ge b_i(\mathcal M_h)$ where $C_i$ denotes
	the number of critical points of index $i$ and where $b_i$ is the $i$-th Betti number of $ \mathcal M_h $.  
 There are $2^{\beta}$ criticial
	points by Lemma \ref{lem-symmetry-count}, and the sum of the Betti numbers of $ \mathcal M_h $, a $ \beta $-dimensional torus, is also $2^{\beta}$. So $C_i = b_i=\binom{\beta}{i}$  for all $i$. 
	
	Assuming parts (\ref{item-nonvanishing}) and (\ref{item-nodal-Morse}), the proof of part (\ref{item-average}) is a simple computation.
	In Part (\ref{item-average}) we assume all the critical points of $\lambda_{k}$ correspond to simple eigenvalues with nowhere-vanishing eigenvectors, which means that all critical points are nondegenerate and are symmetry points, by Part (\ref{item-nonvanishing}). Part
	(\ref{item-nodal-Morse}) says that in such cases the nodal surplus equals the Morse index.  Therefore, the average nodal surplus is
	\begin{align*}
		P(\mathcal S(h))_s &= 2^{-|E|} \sum_{h' \in \mathcal S(h)} P(h')_s\\
		&= \frac{2^{-|E|}}{n} \sum_{h'\in \mathcal S(h)} \#\big\{ k\le n :\ \text{index } (\lambda_k(h')) = s \big \} \\
		&=  \frac{2^{-|E|}}{n} \sum_{k=1}^{n} \# \big\{ h' \in \mathcal S(h): \ \text{index }(\lambda_k(h')) = s \big\} .\end{align*}
	{Using Lemma \ref{lem-symmetry-count} the number inside the parenthesis can be expressed on the quotient $\mathcal M_h$}
	\begin{align*}
		P(\mathcal S(h))_s &= \frac{2^{-\beta}}{n} \sum_{k=1}^{n} \#\big\{ [h'] \in \left[\mathcal S(h)\right]: \ \text{index }(\lambda_k([h'])) = s  \big\}\\
		&= 2^{-\beta} \left( \begin{matrix} \beta \\ s \end{matrix} \right).\end{align*}
	because, by part (\ref{item-perfect}), the number in the parentheses is independent of $k$.
	
	For part (\ref{item-equal-diagonals}), by subtracting a multiple of the identity we may assume the diagonal entries of $h$ are all zero.
	Let $\alpha_{\pi}\in\mathcal{A}(G) $ be properly supported on $G$, with $\pm \pi$ on the non-zero entries. For any $h_{\alpha}=\alpha*h \in \TT_h$ the element $-h_{\alpha} = (\alpha+\alpha_{\pi})*h \in \TT_h$ is also in the same torus but the order of the eigenvalues is reversed, $\lambda_{k}(h_{\alpha})=\lambda_{n-k}(-h_{\alpha})$. 
	This results in an inversion that sends every critical point of $\lambda_{k}$ with index $s$, to a critical point of $\lambda_{n-k}$ with index $\beta-s$. 
	When averaged it gives the needed symmetry around $\beta/2$.

	\subsection{} In this paragraph we prove parts (\ref{item-more-generally}) and (\ref{item-nonvanishing})
	of Theorem \ref{thm-index}. Let $\Tilde{h} \in \TT_{h}$ and suppose that $\lambda_k(\Tilde{h})$ has multiplicity one. $\lambda_{k}$ is analytic in a $\TT_{h}$ neighborhood of $\Tilde{h}$, and we ask when is it a critical point. To ease notation, for this paragraph only, we replace $\Tilde{h}$ by $h$ so that $h$ is now Hermitian rather than real symmetric.  
	Fix a direction $\alpha \in T_{0}(\TT(G))=\mathcal A(G)$ and consider the one-parameter perturbation of $h$ 
	in that direction $h_t = (t\alpha_{0})*h$ for small $ t\in(-\epsilon,\epsilon) $ so that $\dot{h}_{rs} = i\alpha_{rs}h_{rs}$. 
	
	Since $\lambda_k$ is simple,   we get analytic functions $ v(t)\in\CC^{n} $ and $ \lambda_{k}(t)\in\CC $, 
	such that for all $ t\in(-\epsilon,\epsilon) $, the vector $ v(t) $ is normalized and satisfies  $ h_t v(t) = \lambda_k(t)v(t)$. 
	Using Leibniz ``dot" notation for derivative with respect to $t$ at $ t=0 $ we have
	\begin{equation}\label{eqn-dot-formula} \dot{h}v + h \dot{v} = \dot{\lambda} v + \lambda \dot{v}.\end{equation}
	Taking the inner product with $v=v(0)$, using that $h$ is self adjoint, gives 
	\[\langle\nabla\lambda(h),\alpha\rangle:=\dot{\lambda} = \langle v, \dot{h} v \rangle=\sum_{r\sim s}i\alpha_{rs}(h_{rs}\bar{v}_{r}v_{s}-\bar{h}_{rs}v_{r}\bar{v}_{s}),\]
 so $(\nabla\lambda(h))_{rs}=i(h_{rs}\bar{v}_{r}v_{s}-\bar{h}_{rs}v_{r}\bar{v}_{s})=2\Im(\bar{h}_{rs}v_{r}\bar{v}_{s})$ for all $r\sim s$. Therefore, $h$ is a critical point if and only if $h_{rs}\bar{v}_{r}v_{s}\in\RR$. 
	
	
	Assume $\nabla\lambda(h) = 0$.  If $v_r \ne 0$ set $v_r = R_r e^{i \theta_r}$ with $R_r > 0$, otherwise set $\theta_r = 0.$
	Then $h_{rs}R_rR_se^{i(\theta_s-\theta_r)}$ is real for all $r,s$.
	Set $e^{i\theta} = \diag(e^{i\theta_1}, e^{i\theta_2},\cdots, e^{i\theta_n})$.
	Then $h' = e^{-i\theta} h e^{i\theta}$ is gauge equivalent to $h$, and $h'_{rs}=h_{rs}e^{i(\theta_s-\theta_r)}$ 
	is real whenever $\bar v_r v_s \ne 0$.  In particular,
	if $v$ is nowhere vanishing then $h'$ is a symmetry point.  If $h'_{rs} \notin \RR$ then either $v_r=0$ or
	$v_s= 0$.  This completes the proof.
	
	\subsection{}  In the following paragraphs we prove part (\ref{item-nodal-Morse}) of Theorem \ref{thm-index}.
	The result was proven by Berkolaiko \cite{Berkolaiko} and Colin de Verdi\`ere \cite{CdeV2} for real symmetric $ h $ with non-positive off-diagonal entries. Both proofs extend to any real symmetric $ h $, as the authors noted, if one defines the nodal count as in equation (\ref{eqn-nodal-count}). 
	We reorganize the proof of \cite{CdeV2} and present it here for completeness and for later use. Now let $h_{\alpha}\in \TT_{h}\subset\mathcal H(G)$ be an element whose equivalence class is a symmetry point, i.e., $h_{\alpha}$ is gauge equivalent to a real symmetric matrix. For convenience we change the notation slightly, using $h$ instead of $h_{\alpha}$, so suppose $h \in \mathcal H(G)$ is Hermitian properly supported on $G$, which is a critical point of $\lambda_{k}$, with a simple eigenvalue $\lambda:=\lambda_{k}(h)$ and a nowhere-vanishing eigenvector $v$. By Theorem \ref{thm-index},
	\begin{equation}\label{eqn-real}h_{rs}\bar v_r v_s \in \RR \text{ for all }\ r \sim s . \end{equation}
	Let $\ind(Q)$ denote the number of negative eigenvalues of a quadratic form $Q$ and use $\Hess(F)$ for the Hessian of a function
	$F:\TT(G)\to \RR$, evaluated at $\alpha = 0$.  It is a quadratic form on the tangent space $T_0\TT(G) = \mathcal A(G)$.
	
	Define $\mu:\TT(G) \to \RR$ by $\mu(\alpha) = \lambda_k(\alpha * h)$. Since $\mu(\alpha + d\theta) = \mu(\alpha)$ for all
	$\theta \in \RR^n$ it follows that
	the Morse index of $\lambda_k$ at the point $h \in \TT(G)$ is  
	\[ \ind(\lambda_k)(h) = \ind(\Hess(\mu)) = \ind(\Hess(\mu)|V)\]
	for {\em any} complement $V \oplus d\RR^n = \mathcal A(G)$.  The trick (\cite{CdeV2})  is to define $F: \TT(G) \to \RR$ by
	\begin{equation}\label{eqn-F}
		F(\alpha) = \langle v, (\alpha*h-\lambda)v\rangle = \sum_{r \sim s} \bar v_r e^{i \alpha_{rs}} h_{rs} v_s + \sum_r |v_r|^2h_{rr}-\lambda\end{equation}
	{\em where  $\lambda= \lambda_k(h)$ and $v$ are constant},  and show that
	\begin{enumerate}
		\item $\alpha = 0$ is a nondegenerate critical point of $F$ with $\ind(\Hess(F)) = \phi(h,k)$
		\item $\ind(\Hess(F)|d\RR^n) = k-1$
		\item $\ind(\Hess(\mu)|V) = \ind(\Hess(F)|V)$ where $V$ is now chosen to be the orthogonal complement of $d\RR^n$ with respect to $\Hess(F)$:
		\begin{equation}\label{eqn-orthogonal}
			\alpha \in V \iff \langle \alpha, \Hess(F)(d\theta)\rangle = 0 \text{ for all } \theta \in \RR^n.\end{equation}
	\end{enumerate}
	These three steps complete the proof of Theorem \ref{thm-index} (\ref{item-nodal-Morse}) because they give:
	\begin{align*}    
	\ind(\lambda_k)(h) =\ind(\Hess(F)|V) &= \ind(\Hess(F)) - \ind(\Hess(F)|d\RR^n)\\
     &= \phi(h,k) - (k-1).\end{align*}

	\subsection{Step 0} \label{subsec-Step0} 
	To compute $ \Hess(\mu) $, namely the Hessian of $\lambda_k(h_{\alpha})$ at $ \alpha=0 $, let $\gamma,\delta \in \mathcal A_{n}(G)$, set $h(s,t) = (s\gamma + t \delta)*h$ and set
	$\mu(s,t) = \lambda_k(h(s,t))$ with corresponding normalized eigenvector $ v(s,t) $.  Using dot and prime to denote derivatives in $s, t$ at $s= 0, t=0$ respectively, we claim that
	\begin{equation}\label{Hess_calc}
		\langle \gamma, \Hess(\mu)\delta \rangle = 2 \Re\left(\langle v', \dot h v \rangle \right) + \langle \gamma, \Hess(F)\delta \rangle .
	\end{equation}
	Differentiating $\mu(s,t) -\lambda I = \langle v(s,t), (h(s,t)-\lambda)v(s,t) \rangle$ gives;
	\begin{align*}
		\dot {\mu}' = &\langle \dot{v}', (h - \lambda I) v\rangle+\langle v', \dot{h} v\rangle+\langle v', (h - \lambda I) \dot{v}\rangle\\
		&+\langle \dot{v}, h' v\rangle+\langle v, \dot{h}' v\rangle+\langle v, h' \dot{v}\rangle\\
		&+\langle \dot{v}, (h - \lambda I) v'\rangle+\langle v, \dot{h}v'\rangle+\langle v, (h - \lambda I) \dot{v}'\rangle\\
		=&2\Re[\langle v', \dot{h}v\rangle+\langle \dot{v}, (h - \lambda I) v'\rangle+\langle \dot{v}, h'v\rangle]
		+\langle v,\dot{h}'v\rangle,
	\end{align*}
	where $ \langle \dot{v}', (h - \lambda I) v\rangle $ and $ \langle v, (h - \lambda I) \dot{v}'\rangle $ vanish because $ v\in\ker(h-\lambda I) $. Furthermore, the criticality condition $\mu' = 0$ applied to (\ref{eqn-dot-formula}) gives 
	\begin{equation}\label{eqn-first-derivative}
		(h - \lambda I)v' = -h' v,
	\end{equation}
	so the term $ \langle \dot{v}, (h - \lambda I) v'\rangle+\langle \dot{v}, h'v\rangle $ vanishes and we are left with  $\dot {\mu}' =2\Re(\langle v', \dot h v \rangle ) + \langle v, \dot h'v \rangle$. We are done, since $ \dot {\mu}'=\langle \gamma, \Hess(\mu)\delta \rangle  $ and $ \langle v, \dot h'v \rangle=\langle \gamma, \Hess(F)\delta \rangle  $.
	
	\subsection{Step 1}  Calculate:
	\[
	\frac{\partial F}{\partial \alpha_{rs}} = i \left(\bar v_r e^{i \alpha_{rs}}h_{rs} v_s - \bar v_s e^{-i\alpha_{rs}}\bar h_{rs} v_r\right)\]
	which vanishes at $\alpha = 0$ by (\ref{eqn-real}), and
	\begin{equation}\label{eqn-Hessian} \frac{\partial^2F}{\partial^2\alpha_{rs}}(0) = - (\bar v_r h_{rs}v_s + \bar v_s \bar h_{rs}v_r)=-2\bar v_r h_{rs}v_s\end{equation}
	which is real and non-zero for $r\sim s$,   and all other second derivatives vanish.  Therefore $\Hess(F)$ is nondegenerate with index
	\[ \ind(\Hess(F)) = \# \big\{ r\sim s, r<s|\ \bar v_rh_{rs}v_s >0\big\} = \phi(h,k).\]
	
	\subsection{Step 2} \label{subsec-Step2}
	For (small) $t \in \RR$ and $\theta \in \RR^n$ we will find the second derivative of
	\[ F(d(t\theta)) = \langle v, (d(t\theta)*h - \lambda)v \rangle = \langle e^{-i t\Theta}v, (h - \lambda) e^{-i t \Theta}v \rangle\]
	by equation (\ref{eqn-gauge-transformation}), where $\Theta = \diag(\theta_1,\theta_2, \cdots, \theta_n)$ so $e^{i\Theta} = \diag(e^{i\theta_1}, e^{i \theta_2},
	\cdots, e^{i \theta_n})$ (which was formerly denoted $e^{i\theta}$).  Then
	\begin{equation}\label{eqn-dFdt}
		\frac{d}{dt}F(d(t\theta)) = i \langle \Theta e^{-it\Theta}v, (h-\lambda)e^{-it\Theta}v \rangle - i \langle e^{-it\Theta}v, (h-\lambda)\Theta e^{-it\Theta}v \rangle.\end{equation}
	\[ \frac{d^2}{dt^2}F(d(t\theta))\big\vert_{t=0} =
	- \langle \Theta^2v, (h-\lambda)v \rangle - \langle v, (h-\lambda)\Theta^2v \rangle + 2 \langle \Theta v, (h-\lambda)\Theta v \rangle.\]
	
	The first two terms vanish. Using $M_{v} = \diag(v_1,v_2, \cdots, v_n)$ we get
	\[\langle d\theta, \Hess(F)(d\theta) \rangle=2 \langle \Theta v, (h-\lambda)\Theta v \rangle=2 \langle \theta ,M_{v}^{*} (h-\lambda)M_{v}\theta \rangle,\]
	for all $\theta\in\RR^{n}$. According to \eqref{eqn-real}, $2M_{v}^{*} (h-\lambda)M_{v}$ is real and is therefore equal to $\Hess(F)|d\RR^{n}$ as these are real symmetric matrices with equal quadratic forms. The matrix $M_{v}$ is invertible since $v$ is nowhere-vanishing. Two conclusions follow:
	\begin{enumerate}
		\item Assume $\alpha=d\theta\in V \cap d\RR^n$. Then $M_{v}\theta\in\ker \left(h-\lambda\right)$ and so $M_{v}\theta\propto v$ since $\lambda$ is simple. Then $\theta$ is constant, so $d\theta=0$. We conclude that 
		\begin{equation}\label{eqn-VplusRn}V \oplus d\RR^n = \mathcal A(G).\end{equation} 
		\item $\Hess(F)|d\RR^n$ and $h-\lambda$ have the same number of negative
		eigenvalues, so
		\[\ind(\Hess(F)|d\RR^n)=k-1.\]
	\end{enumerate}

	\subsection{Step 3} For any $ \theta\in\RR^{n} $, the derivative of $ v $ in direction $ d\theta $ is $ v'=i\Theta v $ and $ \Hess(\mu)d\theta=0 $ due to gauge invariance. Let $\partial_{\alpha}h$ stand for the derivative of $h$ in direction $ \alpha\in\mathcal{A}(G) $, so that for any $ \delta=d\theta $ equation (\ref{Hess_calc}) gives
	\begin{equation*}
		\langle \alpha, \Hess(F)d\theta \rangle= -2 \Re\left(\langle v', \partial_{\alpha}h v \rangle \right)=2 \Im\left(\langle\Theta v, \partial_{\alpha}h v \rangle \right). 
	\end{equation*}
	It follows from (\ref{eqn-real}) that $\langle \Theta v, \partial_{\alpha}hv \rangle$ is purely imaginary, so $ \alpha\in V $ if and only if $ \langle\Theta v, \partial_{\alpha}h v \rangle  $ vanish for all real diagonal $ \Theta $. Since $ v $ is nowhere-vanishing, 
	\begin{equation}\label{eqn-hdotv} \alpha \in V \iff \partial_{\alpha}h v = 0.\end{equation}
	Consequently, equation (\ref{Hess_calc}) shows that 
	$\Hess(\mu)$ and $\Hess(F)$ agree on $V$.\qed

	\section{Proof of Theorem \ref{thm-linkage}}\label{sec-proofs2}
\subsection{The critical set \texorpdfstring{$F'$}{F1}}	
Recalling the notations of \S \ref{subsec-setting},  the graph $ G $ has $ n+1 $ vertices labeled $ 0,1,\ldots,n $.  The set $ E_{0} $ is the set of edges connected to $ 0 $.  
The graph $ H $ is the induced graph on the non-zero vertices. The torus of perturbations and its tangent space decompose: 
	\[\mathcal{A}(G)=\mathcal{A}(E_{0})\oplus\mathcal{A}(H),\quad \TT(G)=\TT(E_{0})\oplus\TT(H)\]
and $h_{\alpha} = \alpha*h=\left( \begin{smallmatrix} a & b_{\alpha} \\b_{\alpha}^* & D_{\alpha} \end{smallmatrix} \right)$ is an exceptional critical point of $\lambda_k$ with simple eigenvalue $\lambda=\lambda_k(h_{\alpha})$ and real eigenvector $v = (0,v')$ with $v_{0}=0$ and $v_{r}>0$ for $r>0$. As discussed in \S \ref{subsec-setting}, 
$D_{\alpha}$ must be real, so it is a signing of $D$. By replacing $h\in \mathcal {S}(G)$ with a signing of $h$ (if necessary)
 we may assume $D_{\alpha}=D$ and $\alpha \in \mathcal A(E_0)$, so $h_{\alpha}v=\lambda v$ becomes $Dv'=\lambda v'$ and
\[ b_{\alpha}.v' = \sum _{r \in E_0} b_r v_r e^{i \alpha_{0r}}=0.\]

Recall that $F$ denotes the union of the connected
components of the critical set of $\lambda_k$ in $\mathcal M_h$ that contain $[h_{\alpha}]$ and $[\bar h_{\alpha}]$. 
In \S \ref{subsec-diffeomorphism} and  \S \ref{subsec-FFprime} we will prove that $F=F'$ where
	\begin{equation}\label{eqn-Fprime}
	F':=\left\{[h']\in\mathcal M_h : h'v = \lambda v \text{ and }\exists \gamma \in \mathcal A(E_0) \text{ such that }
	h' = \gamma*h_{\alpha}\right\}.
    \end{equation} 
Observe that {\em the set $F'$ is closed under complex conjugation} because $\lambda$ and $v$ are real, and  
$\overline{\gamma *h_{\alpha}} = (-\gamma-2\alpha)*{h}_{\alpha}$
 for any $\gamma \in \mathcal A(E_0)$. 
   Moreover, {\em the set $F'$ consists of critical points of $\lambda_k$}: 
    since $h_{\alpha}$ is critical and $\lambda_k$ is simple,  Theorem \ref{thm-index} part (3) implies $ v_{s}(h_{\alpha})_{rs}v_{r} $ is real for every $ r\sim s $. 
    Since $\lambda=\lambda_{k}(h')$ is simple for any $h' = \gamma*h_{\alpha}\in F'$, then $ v_{s}h'_{rs}v_{r}=v_{s}(h_{\alpha})_{rs}v_{r} $ is real for every $ r\sim s $ (since $v_0 = 0$) hence $[h']$ is also a critical point.    
    
    \subsection{The diffemorphism between \texorpdfstring{$F'$}{F1} and \texorpdfstring{$\Theta_L$}{Th1}}
    \label{subsec-diffeomorphism}
 	Define $ \Phi:\TT(E_{0})*h_{\alpha}\overset{\cong}{\longrightarrow} (S^{1})^{E_{0}} $ by 
\[ \Phi(\gamma*h_{\alpha})_{r}= \begin{cases} e^{i\gamma_{0r}}e^{i\alpha_{0r}} &\text{ if } b_{r}>0 \\  -e^{i\gamma_{0r}}e^{i\alpha_{0r}} &\text{ if } b_{r}<0 \end{cases}\]
for any $\gamma \in \mathcal A(E_0)$.  An element $ h'\in \TT(E_{0})*h_{\alpha}$ satisfies $h'v=\lambda v $ (with the same $ \lambda $ and $ v=(0,v') $) if and only if $ \Phi(h')=(e^{i\theta_{r}})_{r\in E_{0}} $ is a solution to the planar linkage equation $ \sum_{r\in E_{0}}e^{i\theta_{r}}M_{r}=0 $ with $M_{r}:=|b_{r}v_{r}|$.  
The normalizer of $\TT(E_0)*h_{\alpha}$ in the gauge group is the zeroth coordinate $\TT^{(0)}:=\{(x,0,0,\ldots,0)\in\TT^{n+1}\ :\ x\in S^1\}$, cf. equation (\ref{eqn-diag-d}).
Therefore the diffeomorphism $ \Phi $ passes to the quotient,
 \[ \Phi:(\TT(E_{0})*h_{\alpha})\sslash\TT^{(0)}\overset{\cong}{\longrightarrow} (S^{1})^{E_{0}}/S^{1} \] 
with $\Phi (F')=\Theta_L$.
By \cite{Farber} p. 78, the set $ F' $ is a smooth manifold, closed under complex conjugation, and either it is connected or it has two connected components that are exchanged by complex conjugation.  Moreover, it consists of critical points, so $F' \subset F$.  (The reverse inclusion
is proven in \S \ref{subsec-FFprime}.)
		
	\subsection{Gauge transformations on \texorpdfstring{$G$}{G}, \texorpdfstring{$H$}{H} and 
 \texorpdfstring{$E_{0}$}{E0}} 
 Decompose the space of functions on $G$, $\RR^{G}=\RR^{n+1}$, into $\RR^{n+1}=\RR^{(0)}\oplus\RR^{H}$, where $\RR^{H}:=\{(0,x)\in\RR^{n+1}\ :\ x\in\RR^{n}\}\cong \RR^{n}$. The coboundary differential on $H$ is denoted
	\begin{equation}\label{eqn-differential H}
		d_{H}: \RR^H \to \mathcal A(H)\subset\mathcal{A}(G);\quad (d_Hf)_{rs} = 
		\begin{cases} f(s) - f(r) &\text{if } r \sim s\ \text{and}\  r,s\ge1 \\ 0 &\text{otherwise.}\end{cases}\end{equation}
 The image $ d_{H}\RR^{H} $ is the projection of $d\RR^{H}$ into $ \mathcal A(H) $ and in fact, for any $f\in\RR^{H}$,
\begin{equation}\label{eqn-dGdH} df	= d_Hf + \sum_{r\in E_{0}
} f(r)J(r,0)\end{equation}
where $J(r,0) \in \mathcal A(E_0)$ is the antisymmetric matrix with $J(r,0)_{r0} = 1$, $J(r,0)_{0r} = -1$ and all other entries are $0$. Let $\textbf{1}_{H}=(0,1,1,\ldots,1)\in\RR^H$ denote the constant vector on $H$ and let $\textbf{1}_{E_{0}}= \sum_{r\in E_{0}
} J(r,0)\in \mathcal A(E_0)$. It is easy to verify that
\begin{equation}\label{eqn-ROne}
	 \mathrm{span}_{\RR}(\textbf{1}_{E_0}) = d\RR^{(0)} = d(\mathrm{span}_{\RR}(\textbf{1}_{H}))= (d\RR^H)\cap \mathcal A(E_0) = (d\RR^{n+1})\cap \mathcal A(E_0).\end{equation}

Since $D$ is properly supported on $H$ and $\lambda=\lambda_k'(D)$ is a simple eigenvalue of $D$ with a nowhere-vanishing eigenvector $v'$, then $H$ is connected. Theorem \ref{thm-index} gives a decomposition
	\begin{equation}\label{eqn-sum1}\mathcal A(H) = V_H \oplus d_{H}\RR^H\end{equation}
 and $V_H$ can be described in terms of directional derivatives, according to \S \ref{subsec-Step2},   
	\begin{equation}\label{eqn-VH}
		V_H = \{\alpha_{H}\in\mathcal{A}(H)\ :\ \left(\partial_{\alpha_{H}}h_{\alpha}\right)v=0 \},\qquad \partial_{\alpha_{H}}h_{\alpha}:=\frac{d}{dt}\left(t\alpha_{H}*h_{\alpha}\right)|_{t=0}.
	\end{equation}
Let $\mathcal A_0(E_0)$ denote the orthogonal complement to $\mathrm{span}_{\RR}(\textbf{1}_{E_0})$, 
\[\mathcal{A}_{0}(E_{0}):=\left\{\gamma\in \mathcal A(E_{0}): \sum_{r \in E_0}\gamma_{0r}=0\right\}.\]
Let $\pi: \mathcal{A}(G) \to \mathcal{A}(G)/d\RR^{n+1} \cong H^1(G;\RR)$ 
denote the quotient by $d\Omega^0(G) = d\RR^{n+1}$, cf. equation (\ref{eqn-differential}).

\begin{lem}\label{lem-tangent space decomposition}
The space $ \mathcal{A}(G) $ decomposes as a direct sum
\begin{equation}\label{eqn-AVR-decomp}\mathcal{A}(G)=\mathcal A_{0}(E_{0})\oplus V_{H}\oplus d\RR^{n+1}.\end{equation}
In particular, $ \mathcal A_0(E_0) \oplus V_H\overset{\cong}{\longrightarrow}\pi(\mathcal{A}(G))  $ and
$\mathcal A_0(E_0) \overset{\cong}{\longrightarrow} \pi(\mathcal A(E_0))$.
\end{lem}

\begin{proof}
It follows from (\ref{eqn-dGdH}) that $d_H\RR^n \subset d\RR^{n+1} + \mathcal A(E_0)$.  Using (\ref{eqn-ROne}) and \eqref{eqn-sum1},
\[ \mathcal A_0(E_0) + V_H + d\RR^{n+1} = \mathcal A(E_0) + V_H + d\RR^{n+1} \supset \mathcal A(E_0) + (V_H + d_H\RR^n) = \mathcal{A}(G)\]
so $\mathcal{A}(G)$ is spanned by the  sum on the left side.  On the other hand, the sum on the left hand side is a direct sum
because the sum of the dimensions of the vector spaces is
\[ (|E_0|-1) + (\beta_H) + n = (|E_0|-1) + (|E_H| - n +1) + n = |E_G| = \dim(\mathcal{A}(G)).\qedhere\]
\end{proof}

\subsection{The tangent space to \texorpdfstring{$F'$}{F1}}
	 Consider the preimage of $ F' $ in $ \TT(E_{0})*h_{\alpha} $,
	\begin{align*}
		\widehat{F}':= & \{\gamma*h_{\alpha} : \gamma\in\TT(E_{0}),\ [\gamma*h_{\alpha}]\in F'\}\\
		= & \left\{\gamma*h_{\alpha} : \gamma\in\TT(E_{0}) \text{ and } \sum_{r \in E_0}e^{i\gamma_{0r}}(b_{\alpha})_{r}v_{r}=0\right\}.
	\end{align*}
Differentiate and use the  identification  $ T_{h_{\alpha}}\TT_{h} = \mathcal{A}(G) $ 
to obtain the tangent space
\begin{equation}\label{eqn-tangent to F hat}
	T_{h_{\alpha}}\widehat{F}'\cong\left\{\gamma\in\mathcal{A}(E_{0}) : \sum_{r \in E_0}\gamma_{0r}(b_{\alpha})_{r}v_{r}=0\right\} \subset \mathcal A(E_0).
\end{equation}
By Lemma \ref{lem-tangent space decomposition} the quotient projection $\pi$ takes  $\mathcal L:=T_{h_{\alpha}}\widehat{F}' \cap \mathcal A_0(E_0)$ isomorphically to
 $T_{[h_{\alpha}]}F'$, that is,
\begin{equation}\label{eqn-L}
	\mathcal{L}=\left\{\gamma\in\mathcal{A}_{0}(E_{0}) : \sum_{r \in E_0}\gamma_{0r}(b_{\alpha})_{r}v_{r}=0\right\}
	\cong \pi(\mathcal{L})=T_{[h_{\alpha}]}F'.
\end{equation}

	Let $ \RR_{0}^{E_{0}} $ be the space of mean zero elements of $ \RR^{E_{0}} $, which we identify with $\mathcal A_0(E_0)$. 
Let $ \textbf{x}\in\RR_{0}^{E_{0}} $ and $ \textbf{y}\in \RR_{0}^{E_{0}} $ such that $ (b_{\alpha})_{r}v_{r}=\textbf{x}_{r}+i\textbf{y}_{r} $ for all $ r\in E_{0} $. Then 
	\[\mathcal{L}=\{\gamma\in\mathcal{A}_{0}(E_{0}) : \gamma\cdot\textbf{x}=0\text{ and } \gamma\cdot\textbf{y}=0\}.\]

\subsection {The Hessian of \texorpdfstring{$\lambda$}{la}}  
Since $d\RR^{n+1}$ acts by gauge transformations, the quadratic form $\Hess(\lambda_k)$ 
at the critical point $h_{\alpha}$, expressed with respect to the decomposition 
(\ref{eqn-AVR-decomp}) has the following form,
\[\Hess \lambda=\left(\begin{matrix}
		A & C & 0\\
		C^* & B & 0\\
		  0 & 0 & 0.
	\end{matrix}\right).\]
We will show that $C=0$ and $\det(B) \ne 0$ with $\mathrm{ind}(B) = \Phi(D,k') - (k'-1)$.

Using the notation $ \partial_{\gamma} $ for the directional derivative in direction $ \gamma $ and $ \partial^{2}_{\delta,\gamma} $ for the second derivative in direction $ \gamma $ and then in direction $ \delta $, equation \eqref{Hess_calc} states that
\[\langle\gamma, \Hess \lambda \delta \rangle=2\Re[\langle\partial_\gamma v, (\partial_{\delta}h_{\alpha})v\rangle]+\langle v, (\partial^{2}_{\gamma,\delta}h_{\alpha})v\rangle.\] 
If $ \gamma $ is supported on $ E_{0} $ and $ \delta $ is supported on $ H $, then $ \partial^{2}_{\gamma,\delta}h_{\alpha}=0 $. If $ \delta\in V_{H} $, then $ (\partial_{\delta}h_{\alpha})v=0 $ according to \eqref{eqn-VH}. We conclude that $ \langle\gamma, \Hess \lambda \delta \rangle=0 $ when $ \delta\in V_{h} $ and $ \gamma\in\mathcal{A}_{0}(E_{0}) $. Namely, $ C=0 $.

Now consider the block $ B= \Hess \lambda|_{V_{H}}$.  Let $ \Hess F $ be the Hessian of the function $ F(\delta):=\langle v', (\delta*D) v'\rangle $ for $ \delta\in \TT(H) $ evaluated at $ \delta=0 $. Since $ D $ is real, then it is a critical point of $ \lambda_{k'} $. Since $ \lambda=\lambda_{k'}(D) $ is simple with a nowhere-vanishing eigenvector $ v'$, Theorem \ref{thm-index} implies  the restriction $ \Hess F|_{V_{H}} $ is non-degenerate and has $ \Phi(D,k')-(k'-1) $ negative eigenvalues.
Suppose $ \gamma,\delta\in V_{H} $, then $(\partial_{\delta}h_{\alpha})v=0$ and $ (\partial_{\gamma}h_{\alpha})v=0 $ due to \eqref{eqn-VH}, and equation \eqref{Hess_calc} gives 
\[\langle\gamma,B\delta\rangle=\langle v, (\partial^{2}_{\gamma,\delta}h_{\alpha})v\rangle=\langle\gamma,\Hess F\delta\rangle.\]
Therefore $B=\Hess F|_{V_{H}}$. That is,
\begin{equation}\label{eqn-B}
	\det(B)\ne0\text{ and }\mathrm{ind}(B)=\Phi(D,k')-(k'-1).
\end{equation}

\subsection{ The block \texorpdfstring{$ A= \Hess \lambda|_{\mathcal{A}_{0}(E_{0})}$}{A}}	
In this case, for $ \gamma,\delta\in\mathcal{A}_{0}(E_{0}) $ the matrix $ \partial^{2}_{\gamma,\delta}h_{\alpha}$ is supported on $ E_{0} $ so the second term in equation \eqref{Hess_calc} vanishes and we get 
\[\langle\gamma, \Hess \lambda \delta \rangle=2\Re[\langle\partial_\gamma v, (\partial_{\delta}h_{\alpha})v\rangle].\] 
The vector $ (\partial_{\delta}h_{\alpha})v $ is only non-zero at the first coordinate,
\[((\partial_{\delta}h_{\alpha})v)_{0}=i\sum_{r \in E_0}\delta_{0r}(b_{\alpha})_{r}v_{r}=i\delta\cdot\textbf{x}-\delta\cdot\textbf{y}.\]
To calculate $ \partial_\gamma v $,  use \eqref{eqn-first-derivative}: which states
\begin{equation}\label{eqn-first-derivative2}
	(h_{\alpha}-\lambda I)\partial_\gamma v=-(\partial_\gamma h_{\alpha})v.
\end{equation} 
Let $ \psi_{j} $ for $ j=1,2,\ldots,n+1 $ be a choice of orthonormal eigenvectors of $ h_{\alpha} $ corresponding to the ordered eigenvalues. The Moore-Penrose Pseudo-inverse of $ (h_{\alpha}-\lambda I) $ is the matrix 
\[(h_{\alpha}-\lambda I)^{+}:=\sum_{j\ :\ \lambda_{j}(h_{\alpha})\ne\lambda}\frac{1}{\lambda_{j}(h_{\alpha})-\lambda}\psi_{j}\psi_{j}^{*}=\sum_{j\ne k}\frac{1}{\lambda_{j}(h_{\alpha})-\lambda}\psi_{j}\psi_{j}^{*},\]
where in the last equality we used that $ \lambda=\lambda_{k}(h_{\alpha}) $ is simple. By left multiplying \eqref{eqn-first-derivative2} with the matrix $ (h_{\alpha}-\lambda I)^{+} $ (whose kernel is spanned by $ v $) we get 
\[\partial_\gamma v=-(h_{\alpha}-\lambda I)^{+}(\partial_\gamma h_{\alpha})v+\tilde{c}v,\]
for some constant $ \tilde{c} $. Having $ v_{0}=0 $ yields
\begin{align*}
	\langle\partial_\gamma v, (\partial_{\delta}h_{\alpha})v\rangle= & -\overline{(h_{\alpha}-\lambda I)_{00}^{+}\left((\partial_\gamma h_{\alpha})v\right)_{0}}\left((\partial_\delta h_{\alpha})v\right)_{0}\\
	=& c({h_{\alpha}}) \overline{(i\gamma\cdot\textbf{x}-\gamma\cdot\textbf{y})}(i\delta\cdot\textbf{x}-\delta\cdot\textbf{y}), 
\end{align*}
where from (\ref{eqn-ch}),  
\[c({h_{\alpha}})=\sum_{j\ne k}\frac{|\psi_{j}(0)|^{2}}{\lambda-\lambda_{j}(h_{\alpha})}=-(h_{\alpha}-\lambda I)^{+}_{00}.\]
We conclude that
\[\langle\gamma,A\delta\rangle=2c({h_{\alpha}})\left((\gamma\cdot\textbf{x})^{2}+(\delta\cdot\textbf{y})^{2}\right)=\langle\gamma,2c({h_{\alpha}})(\textbf{x}\textbf{x}^{*}+\textbf{y}\textbf{y}^{*})\delta\rangle.\]
Since $ c({h_{\alpha}})\ne0 $ by assumption, then $ A $ has rank two over $ \mathcal{A}_{0}(E_{0}) $. In particular,
\[\ker(A)=\mathcal{L},\quad\text{and}\quad\mathrm{ind}(A)=\begin{cases}
	2 & \text{if}\ c({h_{\alpha}})<0\\
	0 & \text{if}\ c({h_{\alpha}})>0
\end{cases}.\]
Since $ \Hess \lambda|_{\mathcal{A}_{0}(E_{0})\oplus V_{H}}=A\oplus B $, then we conclude that $\ker(\Hess \lambda)=\mathcal{L}\oplus d\RR^{n+1}$ and
\[\mathrm{ind}(\Hess \lambda)=\phi(D_{\alpha},k')-(k'-1)+\begin{cases}
	2 & \text{if}\ c({h_{\alpha}})<0\\
	0 & \text{if}\ c({h_{\alpha}})>0
\end{cases}.\] 

\subsection{\texorpdfstring{$F$}{F} is Morse-Bott} \label{subsec-FFprime}

Recall that the submanifold of critical points $F'\subset\mathcal{M}_{h}$ is a \emph{Morse-Bott} critical submanifold of $\lambda_k$ if at every point $[h']\in F'$, the kernel of $\Hess \lambda_k([h'])$ is exactly the tangent space $T_{[h']}F'$ and the number of negative eigenvalues of $\Hess \lambda_k([h'])$ is constant for all $[h']\in F'$. 
By (\ref{eqn-L}) the kernel condition holds at $[h_{\alpha}]$.  Since $c(h_{\alpha})$ is nonzero and
continuous it does not change sign,  so the Morse-Bott condition holds at every point  $[h']\in F'$.  
 As the kernel of the Hessian at a point $[h']\in F'\subset F$ is $ T_{[h']}F'$ and $\lambda_k$ is constant on $F$, the tangent spaces agree, $ T_{[h']}F'= T_{[h']}F$. 
 Since $F'$ is closed it is a union of connected components of $F$.  It  contains both $[h_{\alpha}]$ and its complex conjugate, so $F=F'$ is Morse-Bott
 and   
\[\mathrm{ind}(F)=\phi(D_{\alpha},k')-(k'-1)+\begin{cases}
	2 & \text{if}\ c({h_{\alpha}})<0\\
	0 & \text{if}\ c({h_{\alpha}})>0
\end{cases}.\qed\]

	\section{Transversality to the strata of \texorpdfstring{$\mathcal H_n$}{Hn}}\label{sec-stratification}
\subsection{The strata} \label{subsec-SmRm}
The vector space $\mathcal H_n$ of Hermitian $n \times n$ matrices is stratified according to the multiplicities of the eigenvalues, as described in \cite{Arnold2}. (See also \cite{Arnold1,  Agrachev, Shapiro}.)
Suppose $h \in \mathcal H_n$ has $k$ {\em distinct} eigenvalues $\mu_1 < \mu_2 < \cdots <\mu_k$.  
Specifying a multiplicity $r(i)$ for the eigenvalue $\mu_i$ determines a stratum $T(r)$, 
consisting of Hermitian matrices with eigenvalues $\mu_i$ and multiplicities $r(i)$.
The multiplicity vector $r$ is an ordered partition of $n$, meaning that $n = \sum_{i=1}^k r(i)$, and every
ordered partition of $n$ determines a stratum. The set of possible eigenvalues for $h$ forms an open set
\[ \RR^k_{<} = \left\{ x \in \RR^k:\ x_1 < x_2 < \cdots < x_k\right\}\]
in $\RR^k$.  The eigenspaces $V_i$ determine a partial flag $V_1 \subset V_1 \oplus V_2 \subset \cdots \subset \CC^n$.
Therefore the stratum $T(r)$ may be canonically identified with the product
\[ P(r)=\mathcal{F}l(r) \times \RR^k_{<}\]
where $\mathcal Fl(r)$ denotes the partial flag manifold of subspaces 
$0 \subset W_1\subset W_2 \subset \cdots \subset \CC^n$ with $\dim(W_k) = \sum_{i=1}^kr(i)$.
This identification endows the stratum $T(r)$ with the canonical structure of an analytic manifold, 
and each eigenvalue $\mu_i:T(r) \to \RR$ is an analytic function.

It is well known \cite{Wikipedia} that $\mathcal Fl(r)$ is isomorphic to the
quotient $U(n)/\prod_{i=1}^k U(r(i))$ of unitary groups, 
so it has dimension $n^2-\sum_{i=1}^k r(i)^2$ from which it
follows that the stratum $T(r)$ has codimension $\sum_{i=1}^k(r(i)^2-1)$ in $\mathcal H_n$.

\subsection{The manifold \texorpdfstring{$S_m(k)$}{Smk}}

For any $h \in \mathcal H_n$ we may label the eigenvalues 
$\lambda_1\le \lambda_2 \le \cdots \le \lambda_n$.    Fix $m, k$. Let 
$V_k(h) = \ker(h-\lambda_k(h).I)$ be the eigenspace with eigenvalue $\lambda_k$.  
Define\footnote{We are grateful to the referee for pointing out an error in our earlier definition of $S_m(k)$}
\begin{equation}\label{eqn-Sm}
S_m(k) = \left\{h \in \mathcal H_n:\ \dim(V_k(h))=m \text{ and }\lambda_{k-1}(h)<\lambda_k(h)  \right\}.\end{equation}
Each $h \in S_m(k)$ has exactly $k-1$ eigenvalues less than $\lambda_k$ and $n-m-k+1$ eigenvalues 
greater than $\lambda_k$.  It is foliated with leaves indexed by $\lambda \in \RR$,
\begin{equation}\label{eqn-Rm} S_m(k,\lambda) = \left\{ h \in S_m(k): \lambda_k(h) =\lambda \right\}. \end{equation}

\begin{lem}\label{lem-Smk}  The set $S_m(k)$ (resp. $S_m(k,\lambda)$) is an analytic manifold of codimension $m^2-1$ 
(resp. codimension $m^2$) in $\mathcal H_n$. The eigenvalue $\lambda_k:S_m(k) \to \RR$ is analytic.
If $r$ is an ordered partition of $n$ then stratum $T(r)$ is the transverse intersection
 \begin{equation}\label{eqn-T(r)}
 S_{r(1)}(1) \cap S_{r(2)}(1+r(1)) \cap S_{r(3)}(1+r(1)+r(2)) \cap \cdots \cap S_{r(k)}(n-r(k)+1).\end{equation}
\end{lem}
\begin{proof}  If $h \in S_m(k)$ then the eigenspace $V = V_k(h)$ is an element of the Grassmann manifold
$G_m(\CC^n)$ of $m$-dimensional subspaces of $\CC^n$. Set $t = \lambda_k(h)$.  The restriction $h|V^{\perp}$
determines an orthogonal decomposition
 $V_k(h)^{\perp} = W_- \oplus W_+$ as the sum of the $<t$ (resp. $>t$) eigenspaces of
$h|V_k(h)^{\perp}$ of dimension $k-1$ and dimension $n-m-k+1$ respectively.  The further restriction
$h|W_-$ lies in the set $\mathcal H(W_-)^{<t}$ of Hermitian operators all of whose eigenvalues are $<t$,
and similarly for $h|W_+$.  Therefore we have parametrized $S_m(k)$ by a double fibration
\begin{diagram}[size=2em]
\mathcal H(W_-)^{<t} \times &\mathcal H(W_+)^{>t} &  & \rInto &  & S_m(k) \\
&\dTo    &&  && \dTo \\
&W_-  & \in\ \ & G_{k-1}(V^{\perp}) &\rInto& E\\
&&& \dTo && \dTo\\
&&& (V,t) & \in\ & G_m(\CC^n) \times \RR
\end{diagram}
where $E$ is the bundle whose fiber over $(V,t)$ is the Grassmannian of $k-1$ dimensional complex subspaces  
$W_-\subset V^{\perp}$.  From this we see that $S_m(k)$ is an analytic manifold and $t = \lambda_k$ is
an analytic function on $S_m(k)$, as a coordinate in the parametrization.

The dimension of $S_m(k)$ may be calculated from the above diagram, 
\[\dim(S_{m}(k))= \dim_{\RR}\left(G_m(\CC^n)\right)+\dim_{\RR}\left(G_{k-1}(V^{\perp})\right)+\dim\left(\mathcal H(W_-)^{<t} \times\mathcal H(W_+)^{>t}\right),\]
and a miraculous cancellation of terms
 gives $\mathrm {codim}(S_m(k)) = m^2-1$.  It is a direct consequence of the definitions that
$T(r)$ is the intersection (\ref{eqn-T(r)}).  The intersection is transversal
because the different factors in (\ref{eqn-T(r)}) involve independent conditions.
\end{proof}

\begin{prop}\label{prop-tangent-space}  Fix $h \in \mathcal H_n$ with eigenvalue $\lambda=\lambda_k(h)$ 
and  eigenspace $V=V_{k}$ of dimension $m$.

\begin{enumerate}
\item[(A)]
The tangent space $T_hS_m(k,\lambda) $ (resp. $T_hS_m(k)$) consists of all tangent vectors $\xi \in T_h 
\mathcal H_n = \mathcal H_n$ such that, as sesquilinear forms\footnote{ cf. \S \ref{subsec-Hermitian}}, 
the restriction  $\xi|V = 0$  (resp.~such that $\xi|V$ is a scalar\footnote{ In fact it is multiplication by the 
directional derivative $\partial_{\xi}(\lambda_k)$.}). With respect to the decomposition 
$\CC^n = V_k \oplus V_k^{\perp}$ it is the subspace of matrices
\begin{equation}\label{eqn-BD} \xi = \left( \begin{matrix} 0 & B \\ B^* & D \end{matrix} \right)\quad  \left( \text{resp. } \xi = 
 \left( \begin{matrix} c.I_{V} & B \\ B^* & D \end{matrix} \right) \right)\end{equation}
where $D \in \mathcal H_{n-k}$, $B \in M_{k \times (n-k)}(\CC)$, $c\in\RR$, and $I_{V}$ the identity on $V$.
\item[(B)]
A submanifold $Q\subset \mathcal H_n$ is transverse to $S_m(k,\lambda)$ (resp. $S_m(k)$) at $h \in Q$ if and only if
the elements $\xi|V$ (resp.~the elements $c.I_{V}+\xi|V$) account for all the Hermitian operators in $\mathcal H(V)$, as $\xi$ varies within $T_{h}Q$ and $c$ within $\RR$.  
\item[(C)]
 The tangent space  $T_hS_m(k,\lambda)$ (resp. $T_hS_m(k)$)  can also be expressed as the set of all $\xi \in \mathcal H_n$ of
the form
\[ \xi = (h-\lambda.I)U + U^*(h-\lambda.I) \text{ with } U \in M_{n \times n}(\CC) \]
(resp. $\xi = (h-\lambda.I)U + U^*.(h-\lambda.I) + c.I_{V}$ with $U \in M_{n \times n}(\CC)$ and $c \in \RR$).
\end{enumerate}\end{prop}

\begin{proof}
Let $\xi$ be a tangent vector to $S_m(k,\lambda)$ at the point $h = h_0$.  Let $h_t \in S_m(k,\lambda)$ 
be a smooth one parameter family with $\xi = \dot h=\frac{d}{dt}h(0)$.  Suppose $u_t \in V \subset \CC^n$
is an eigenvector of $h_t$ with eigenvalue $\lambda$.  Differentiating the eigenvalue equation $hu = \lambda u$
gives $\dot h u + h \dot u = \lambda \dot u$.  Taking the inner product with any $w \in V$ gives 
$\langle w, \xi u \rangle = 0$ which shows that $\xi$ has the form of equation (\ref{eqn-BD}) above.
On the other hand the codimension of the space of matrices (\ref{eqn-BD}) is $m^2$ which equals the
codimension of $S_m(k,\lambda)$ so (\ref{eqn-BD}) describes the full tangent space.  A similar procedure
works for the tanget space to $S_m(k)$.

Part (B) of the proposition is an immediate consequence.

For part (C), using the decomposition $\CC^n = V_k \oplus V_k^{\perp}$, 
the matrix of $T = h-\lambda.I$ is 
\begin{equation}\label{eqn-T} h-\lambda.I = \left( \begin{matrix} 0 & 0 \\ 0 & A \end{matrix} \right)\end{equation}
where $A$ is nonsingular.  Given $U \in M_{n \times n}(\CC)$, we have 
\[U=\left( \begin{matrix} U_{1} & U_{2} \\ U_{3} & U_{4} \end{matrix} \right)\ \Rightarrow\ TU + U^*T=\left( \begin{matrix} 0 & U_{3}^{*}A \\ AU_{3} & AU_{4}+U_{4}^{*}A \end{matrix} \right).\]
Hence, $ TU + U^*T\in T_{h}R_{m}(\lambda)$ as it has the form (\ref{eqn-BD}). Conversely, since $A$ is invertible, $AU_{3}$ and $AU_{4}$ account for all matrices in $ M_{(n-k) \times k}(\CC)$ and $ M_{k \times k}(\CC)$ as $U$ varies within $M_{n \times n}(\CC)$. This proves the case $\xi = \left( \begin{smallmatrix} 0 & B \\ B^*& D \end{smallmatrix} \right)$, and the case of $\xi = \left( \begin{smallmatrix} c.I_{V} & B \\ B^*& D \end{smallmatrix} \right)$
follows.
 \end{proof}
 
A closely related result, concerning transversality with respect to the manifold of matrices with
constant rank, appears in \cite[Chapter 10.5]{Lovasz}.   Although Lov\'asz considers 
 only real symmetric matrices, his proof applies also to Hermitian matrices.

\subsection{Application to graphs}\label{subsec-graphs}
Fix $k, m, n \ge 1$ and let $\lambda\in \RR$.  Recall $S_m(k), S_m(k,\lambda) \subset \mathcal H_n$ from
(\ref{eqn-Sm}), (\ref{eqn-Rm}). 
Let $G$ be a graph on $n$ vertices with a set of edges $E$ and associated spaces $\mathcal S(G) \subset \mathcal H(G) \subset \mathcal H_n$.  
In this section we determine when the inclusion $\mathcal H(G) \to \mathcal H_n$ is transverse to the manifold 
$S_m(k)$ at a point $h$ of their intersection.
The results are used in Corollary \ref{cor-splitting}, the case of multiplicity 2,  to provide sufficient conditions which guarantee that 
the mapping $\TT_h \to \mathcal H_n$ is transverse to $S_2(k)$ locally near $h$.

Let $\mathcal H(\overline{G})$ denote those Hermitian matrices that are supported on the complement of $E$. That is $h\in\mathcal H(\overline{G}) $ if $h^{*}=h$, $h_{rr}=0$, and $h_{rs}=0$ for all $r$ and for any edge $rs$. It is the orthogonal
complement\footnote{$\mathcal H_n$ is equipped with the standard inner product $\langle A,B \rangle:=\mathrm{trace}(A^{*}B)=\mathrm{trace}(AB)$. } in $\mathcal H_n$ to $\mathcal H(G)$.

\begin{prop}\label{prop-Lovasz}
Fix $h \in \mathcal H_n$. Suppose $\lambda_{k-1}(h)<\lambda_{k}(h)$ and the eigenvalue $\lambda:=\lambda_k(h)$ has multiplicity $m$ and eigenspace $V=V_{k}$. The following statements are equivalent.
\begin{enumerate}
    \item
The inclusion $\mathcal H(G) \to \mathcal H_n$ is transverse to $S_m(k)$ at $h$.
\item The inclusion $\mathcal H(G) \to \mathcal H_n$ is  transverse to $S_m(k,\lambda)$ at $h$.
\item $(h-\lambda.I)X \ne 0$
for every nonzero $X \in \mathcal H(\overline{G})$.
\item There exist $\xi_1,\xi_2,\cdots,\xi_N \in \mathcal H(G)$ whose restrictions $\{ \xi_1|V,\xi_2|V,\cdots,\xi_N |V\}$ span (over $\RR$) the space $\mathcal H(V)$ of
Hermitian operators on $V$.
\end{enumerate}
\end{prop}

\begin{proof}
Let $V$ denote the $m$ dimensional eigenspace of $h$ with eigenvalue $\lambda$.
Parts (1) and (2) are equivalent by Proposition \ref{prop-tangent-space} because the tangent space $T_h\mathcal H(G) = \mathcal H(G)$ 
contains the identity matrix $\xi = I$, and $\xi|V = I_V$.

For part (3), as in \cite{Lovasz} \S 10.5.2, the manifolds $\mathcal H(G)$ and $S_m(k,\lambda)$ 
are transverse at $h$ if and only if the orthogonal complements
of their tangent spaces intersect trivially.  The orthogonal complement to $T_{h}S_m(k,\lambda)$ is 
\[ T_hS_m(k,\lambda)^{\perp} = \left\{ X \in \mathcal H_n: (h-\lambda.I)X = 0 \right\} \]
by equations \eqref{eqn-T} and (\ref{eqn-BD}). The orthogonal complement of $T_{h}\mathcal H(G)=\mathcal H(G)$ 
is $\mathcal H(\overline{G})$, so transversality to $S_m(k,\lambda)$ fails if and only if there exists 
$0 \ne X \in \mathcal H(\overline{G})$  such that $(h-\lambda.I)X = 0$.
  
  Part (4) is a restatement of part (B) of Proposition \ref{prop-tangent-space}.
\end{proof}

 \subsection{Example - Graph splitting}  The following example provides some intuition for the definitions in
 \S \ref{subsec-graph-conditions}.  
 Given graphs $H_1, H_2$ of size $n_{1},n_{2}$ respectively. Suppose $h_{1}\in\mathcal{H}(H_{1})$ and $h_{2}\in\mathcal{H}(H_{2})$ have the same simple eigenvalue $\lambda$.  
 Let $\phi_1,\phi_2$ be corresponding eigenfunctions.  Suppose there is a vertex $v_1$ of $H_1$ with $\phi_1(v_1)=0$
 and a vertex $v_2$ of $H_2$ with $\phi_2(v_2) = 0$.  Let $G$ be the graph of size $n=n_{1}+n_{2}-1$ obtained from joining $H_1$, $H_2$ by identifying the vertices $v_1$ and $v_2$. Define $h\in\mathcal{H}(G)$ such that its restrictions to $H_{1},H_{2}$ agree with $h_{1},h_{2}$.  Then $\lambda$ is an eigenvalue of $h$ with 2 dimensional eigenspace $V$ and eigenfunctions
 $\Phi_1 = \phi_1 \times \{0\}$ and $\Phi_2 = \{0\} \times \phi_2$ such that $\langle \Phi_1, \Phi_2 \rangle = 0$.
For any $\xi\in \mathcal H(G)$ we have $\langle \Phi_1, \xi \Phi_2\rangle = 0$ so the elements $\xi|V$ fail
to account for all quadratic forms on $V$, that is, $\mathcal H(G)$ is not transverse
to $S_2(k)$ at the point $h$. In the graph $G$, both eigenfunctions vanish at the 
vertex $v_1= v_2$ so the graph $G$ is ``split" into two pieces by this eigenvalue  $\lambda$.

\subsection{Graph theoretic conditions}\label{subsec-graph-conditions}
Maintain the notation of \S \ref{subsec-SmRm} and \S \ref{subsec-graphs}.  If $u \in \RR^n$ the {\em support} of $u$ is the set $\mathrm{spt}(u)$
of vertices $j$ such that $u_j \ne 0$. If $V \subset \RR^n$ is a subspace, its support is
\[ \mathrm{spt}(V) = \bigcup_{u \in V} \mathrm{spt}(u).\]
If $\lambda$ is an eigenvalue of $h\in\mathcal{H}(G)$ with eigenspace $V=\ker(h-\lambda I)$, 
we say that {\em the eigenspace $V$ splits $G$} if the induced subgraph $G|\spt(V)$ of $G$ on the vertices in $\mathrm{spt}(V)$ is {\em not} connected. 
Given an edge $(rs)$ we say that {\em $V$ projects surjectively onto $(rs)$} if $\{(u_{r},u_{s})\ :\ u\in V\}=\CC^{2}$. 
\begin{thm}\label{prop-transverse}  Suppose $h \in \mathcal H(G)$ has eigenvalue $\lambda=\lambda_k(h)$ and eigenspace $V$.
\begin{enumerate}[ align=right]
	\item[(A)] Suppose the multiplicity of $\lambda_k$ is 2 and  either 
		\begin{enumerate}
		\item There exists an edge $(rs)$ on which $V$ projects surjectively, or
		\item The eigenspace $V$ does not split $G$.
		\end{enumerate}
Then $\mathcal H(G)$ is transverse to $S_2(k)$ at the point $h$.

	\item[(B)] For arbitrary multiplicity $m$, suppose there exist nonzero vectors $u,v \in V$ whose supports are edge-separated, meaning that ${\mathrm{spt}}(u)\cap {\mathrm{spt}}(v)=\emptyset$ and there are no edges between ${\mathrm{spt}}(u)$ 
and $\mathrm{spt}(v)$. Then
$\mathcal H(G)$ is not transverse to $S_m(k)$ at $h$.
\end{enumerate}
\end{thm}

\begin{proof}

 Assume there exists an edge $(rs)$ on which $V$ projects surjectively. In this case, we may choose $u$ and $v$ in $V$
 such that $(u_{r},u_{s})=(1,0)$ and $(v_{r},v_{s})=(0,1)$, so that in the (not necessarily orthonormal) basis $\{u,v\}$ of $V$, for any $\xi\in\mathcal{H}(G)$ with $\xi_{ij}=0$ for all $ij\notin\{rr,rs,sr,ss\}$,
  \[\xi|_{V}:=\left(\begin{matrix}
		\langle u,\xi u\rangle & \langle u,\xi v\rangle\\
		\langle v,\xi u\rangle & \langle v,\xi v\rangle
	\end{matrix}\right)=\left(\begin{matrix}
		\xi_{rr} & \xi_{rs}\\
		\xi_{sr} & \xi_{ss}
	\end{matrix}\right).\]		
So these vectors span $\mathcal H(V)$ verifying part (4) of Proposition \ref{prop-tangent-space}.

For condition (b), we will show that if $V$ does not split $G$, then there must be an edge $(rs)$ on which $V$ projects surjectively. Let $H$ be the induced subgraph on $\mathrm{spt}(V)$ and assume it is connected. Choose a generic basis $\{u,v\}$ of $V$, so that $u(r)\ne 0$ and $v(r)\ne0$ for all $r\in H$.  Assume by contradiction that $ V$ does not project surjectively on any edge $(rs)$. That is, $\frac{u(s)}{u(r)}=\frac{v(s)}{v(r)}$ for every edge $(rs)$  in $H$. Fix an initial vertex $r_{0}\in H$. Any other vertex $s\in H$ is connected to $r_{0}$ by a path in $H$, say $(r_{0},r_{1},r_{2},\ldots,r_{m},s)$, and so   
		\[\frac{v(s)}{v(r_{0})}=\frac{v(r_{1})}{v(r_{0})}\cdot\frac{v(r_{2})}{v(r_{1})}\ldots\frac{v(s)}{v(r_{m})}=\frac{u(r_{1})}{u(r_{0})}
		\cdot\frac{u(r_{2})}{u(r_{1})}\ldots\frac{u(s)}{u(r_{m})}=\frac{u(s)}{u(r_{0})}.\]
Thus $u$ and $v$ are linearly dependent.

For Part (B), given $u,v\in V$ as described in Part(B), the matrix $X=uv^{*}+vu^{*}$ is in $\mathcal{H}_{n}(\overline{G})$ and satisfies $(h-\lambda.I)X=0$, which contradicts 
Proposition \ref{prop-Lovasz} part (3).  
\end{proof}
	
	\begin{cor}\label{cor-splitting}
Let $h \in \mathcal S(G)$ be properly supported, and let $\alpha \in \mathcal A(G)$.  
Let $\TT_h=\TT(G)*h\subset \mathcal H(G)$ be the embedded torus.
Suppose the eigenvalue $\lambda_k$ of $h_{\alpha} = \alpha*h$ has multiplicity 2 with eigenspace that does not split $G$.  Then there is a neighborhood $V \subset \mathcal A(G)$ of $\alpha$ 
and a neighborhood $U \subset \mathcal S(G)$ of $h$ such that for a 
generic\footnote{An open, dense and full measure set.} set of $h' \in U$,
 the embedding map $\TT_{h'} {\hookrightarrow} \mathcal H_n$ takes the open subset 
 \[ V*h' = \left\{ \alpha'*h' : \alpha' \in V \right\} \subset \TT_{h'}\] 
 transversally to the stratum $S_2(k)$.
\end{cor}

\begin{proof}  Consider the composition 
	\[\Phi: \TT(G) \times \mathcal S(G) \overset{*}{\to} \mathcal H(G) 
\overset{j}{\hookrightarrow} \mathcal H_n \]
given by $\Phi(\alpha', h') = \Phi_{h'}(\alpha')= \alpha'*h'$.  The  map $*$ above is surjective, since $h$ is properly supported, with finite fibers, it is an open mapping and a submersion. The non-splitting assumption
implies the embedding map $j$ takes $\mathcal H(G)$ transversally to $S_2(k)$ at the point $h_{\alpha}$ so it takes a
neighborhood $W\subset \mathcal H(G)$ of $h_{\alpha}$ transversally to $S_2(k)$.  Choose the neighborhoods
$V\subset \mathcal A(G)$ and $U \subset \mathcal S(G)$ so that
$V*U \subset W$.  Then $\Phi:V \times U \to \mathcal H_n$ is transverse to $S_2(k)$.   Lemma 
\ref{lem-transversality} implies there exists a dense set of values $h' \in U$ so that
the resulting map $\Phi_{h'}:V \to \mathcal H_n$ is transverse to $S_2(k)$.  But this map is the
composition
\[ V \overset{\cong}{\longrightarrow} V*h' \overset{j}{\longrightarrow}\mathcal H_n.  \qedhere\]
\end{proof}
\subsection{Example - Graphs for which \texorpdfstring{$\TT_{h}$}{Th} is generically transverse to 
\texorpdfstring{$S_{2}$}{S2}} Suppose $G$ is a graph obtained by removing a set of disjoint edges from the complete graph, say $(r_{j},s_{j})$ for $j=1,\ldots,m$ such that 
the vertices $\{r_{1},s_{1},r_{2},s_{2},\ldots\}$ are all distinct. If $h\in\mathcal{H}(G)$ has distinct diagonal elements
(a generic assumption) then the embedded torus $\TT_h$ intersects $S_2(k)$ transversally for every $k$.
To prove this, it suffices by Corollary \ref{cor-splitting} to show, for any $h_{\alpha}\in \TT_h$,
that no multiplicity-two eigenvalue of $h_{\alpha}$ splits $G$.

Assume by contradiction that some $h_{\alpha} \in \TT_h$ has a multiplicity two eigenvalue $\lambda=\lambda_k(h_{\alpha})$
with eigenspace $V$ such that the induced graph $G|\spt(V)$ is disconnected.  By the construction of $G$
this means that $\spt(V)=\{r_{j},s_{j}\}$ for one of the missing edges $(r_{j},s_{j})$. So $\lambda$ is a 
multiplicity-two eigenvalue of the restriction
\[h_{\alpha}|\spt(V)=\begin{pmatrix}
    h_{r_{j},r_{j}} & 0\\ 0 & h_{s_{j},s_{j}}\end{pmatrix}.\]
    This contradicts the assumption that diagonal elements of $h$ are distinct.

\appendix

	\section{Transversality}

 \begin{translem}\label{lem-transversality}
		Let $\Phi:\TT\times  B \to \mathcal H$ be a smooth map between smooth manifolds and suppose this map is transverse to a submanifold 
		$S \subset \mathcal H$.  Then there is a dense set of values $b \in  B$ such that the partial map 
		\[\phi_b: \TT \to \mathcal H\ \text{ given by }\ \phi_b(x) = \Phi(x,b) \]
		is transverse to $S$.  If $\Phi$ is proper and $S \subset \mathcal H$ is closed then this set of values is open in $B$.  If $\Phi, \TT, B, \mathcal H$ and $S$
		are analytic then the set of values $b \in B$ for which transversality of $\phi_b$ fails is a subanalytic subset of $B$ of positive codimension.
	\end{translem}

 \paragraph{Remarks}
 Here, $\TT$ is any finite-dimensional smooth manifold. The symbol $\TT$ is being used to indicate that for our application, $\TT$ is an open subset of the torus $\TT(G)$.
	
  This result says, for example, that two submanifolds of Euclidean space may be made transverse by an arbitrarily small {\em translation}.  The transversality
	lemma is due originally to R. Thom (\cite{Thom}).  The proof described here may be found in (\cite{Guillemin}).

	\begin{proof}
		It suffices to consider the case when $B$ is open in some Euclidean space.
		By assumption, the set $P = \Phi^{-1}(S)$ is a smooth submanifold of $\TT\times B$ and it is easy to check that $b\in B$ is
		a regular value of the projection $\pi:P \to B$ if and only if the partial map $\phi_b:\TT \to \mathcal H$ is transverse to $S$.  But Sard's theorem
		says that the set of non-regular values of $\pi$ has Lebesgue measure zero.
		
		Now assume $S$ is closed and $\Phi$ is proper (i.e., the preimage of a compact set is compact).  To show the set of ``transversal" elements $b \in B$ is open, we show its complement is closed.
		Let $b_i\in B$ be a convergent sequence of points, say $b_i \to b \in B$
		for which there exists points $t_i \in \TT$ such that
		$\phi_{b_i}$ fails to take the tangent space $T_{t_i}\TT$ transversally to $T_{s_i}S$ where $s_i = \Phi(t_i,b_i)$.  Since $\Phi$ is proper, 
		by taking a subsequence if necessary
		we may assume the sequence converge, say $t_i \to t \in \TT$ and therefore $s_i \to s $ for some $s \in \mathcal H$.  
		Since $S$ is closed, we also have $s \in S$.  The failure of transversality is
		a closed condition so  $\phi_b$ fails to take $T_t\TT$ transversally to $T_sS$.      
		
		Finally, if $\Phi, \TT, B, \mathcal H, S$ are analytic then the set of points $(t,b) \in \TT \times B$ for which $\phi_b$ fails to be transverse at $t$ is again
		analytic so its image $Z \subset B$ is a subanalytic subset of $B$.  It has positive codimension, 
		for if $Z$ contains an open set in $B$ then this contradicts
		the assumption that $\Phi$ is transverse to $S$.  
	\end{proof}
 \section{Heuristics for discretization of magnetic Schrödinger operators}\label{sec-magnetic-Schrodinger}
The definition of discrete magnetic operators can be found in \cite{Lieb, CdeVMagnetic} for example, however, we will give here a heuristic explanation for why this is the right discretization for magnetic Schrödinger operators. For simplicity, we consider domains in $\RR^{3}$ so that magnetism can be described using  vector fields: a magnetic field $B$ and magnetic potential $A$ such that $B=\nabla\times A$. 
(The modern approach would consider $A$ and $B$ as a 1-form and 2-forms). \\
The quadratic form of a Schrödinger operator $H=\Delta+V$ on a domain $\Omega\subset\RR^{n}$ is 
\[\langle f,Hf\rangle=\int_{\Omega}\sum_{j=1}^{n}\left(\frac{\partial f(x)}{\partial x_{j}}\right)^{2}+V(x)f(x)^{2}dx,\]
for the relevant class of functions $f$ on $\Omega$.
If we approximate $\frac{\partial f(x)}{\partial x_{j}}$ with $\frac{f(x+\epsilon e_{j})-f(x)}{\epsilon}$, the quadratic form can be written as in \eqref{quadratic form}
\[\sum_{x,y\in\Lambda_{\epsilon}}h_{xy}(f(y)-f(x))^{2}+V(x)f(x)^{2}dx,\qquad h_{xy}=\begin{cases}\frac{1}{\epsilon^{2}} & x\sim y\\
0 & \text{otherwise}\end{cases},\]
where $\Lambda_{\epsilon}\subset \Omega$ is a grid of side length $\epsilon$. Introducing a magnetic field $B$, the operator $H$ is changed to a magnetic Schrödinger operator $H_{A}$ by the rule $\frac{\partial f(x)}{\partial x_{j}}\mapsto \frac{\partial f(x)}{\partial x_{j}}+iA_{j}(x)f(x)$, where $A=(A_{1},A_{2},A_{3}) $ is the {\em magnetic potential} defined (uniquely up to gauge transformations $A\sim A'=A+\nabla g$) by the relation $\nabla\times A=B$. Notice that
\[\frac{\partial f(x)}{\partial x_{j}}+iA_{j}(x)f(x)=\lim_{t\to 0} \frac{(e^{i\int_{x}^{x+te_{j}}A(s)ds}f(x+te_{j}))-f(x) }{t},\]
and the $\epsilon$ discretization of the quadratic form can be written as in \eqref{sesquilinear form} 

\begin{align*}
    \int_{\Omega}&\sum_{j=1}^{3}\left|\frac{\partial f(x)}{\partial x_{j}}+iA_{j}(x)f(x)\right|^{2}+V(x)|f(x)|^{2}dx\\
&\approx\sum_{x,y\in\Lambda_{\epsilon}}h_{xy}\left|f(y)-e^{i\alpha_{xy}}f(x)\right|^{2}+V(x)|f(x)|^{2}dx,\end{align*}
using $\alpha_{xy}=\int_{x}^{x+\epsilon e_{j}}A(s)ds$ when $y=x+\epsilon e_{j}$ and extending it antisymmetrically.

\end{document}